%% file: paper.tex
\documentclass[a4paper,cleveref,thm-restate]{lipics-v2021}
\nolinenumbers
\usepackage{todonotes}
\usepackage{tikz}

\let\tss\ddagger
\newcommand{\fig}[1]{\includegraphics[scale=0.66]{#1}}
\newcommand{\figSm}[1]{\includegraphics[scale=0.59]{#1}}
\newcommand{\atom}[1]{\ensuremath{\mathsf{#1}}}
\newcommand{\motif}[1]{Motif~\textsf{#1}}
\newcommand{\bond}{\,--\,}

\title{Reversible Bond Logic}
\author{Hannah Earley}{Department of Applied Mathematics and Theoretical Physics, University of Cambridge, United Kingdom}{hannah.earley@damtp.cam.ac.uk}{https://orcid.org/0000-0002-6628-2130}{}

\authorrunning{H. Earley}
\Copyright{Hannah Earley}



\acknowledgements{The author would like to thank Gos Micklem, Jim Lathrop, and William Poole for valuable comments and suggestions.}

\begin{document}
\maketitle

\begin{abstract}
  The field of molecular programming allows for the programming of the structure and behavior of matter at the molecular level, even to the point of encoding arbitrary computation.
  However, current approaches tend to be wasteful in terms of monomers, gate complexes, and free energy.
  In response, we present a novel abstract model of molecular programming, Reversible Bond Logic (RBL), which exploits the concepts of reversibility and reversible computing to help address these issues.
  RBL systems permit very general manipulations of arbitrarily complex `molecular' structures, and possess properties such as component reuse, modularity, compositionality.
  We will demonstrate the implementation of a common free-energy currency that can be shared across systems, initially using it to power a biased walker.
  Then we will introduce some basic motifs for the manipulation of structures, which will be used to implement such computational primitives as conditional branching, looping, and subroutines.
  Example programs will include logical negation, and addition and squaring of arbitrarily large numbers.
  As a consequence of reversibility, we will also obtain the inverse programs (subtraction and square-rooting) for free.
  Due to modularity, multiple instances of these computations can occur in parallel without cross-talk.
  Future work aims to further characterize RBL, and develop variants that may be amenable to experimental implementation.
\end{abstract}
\section{Introduction}

Molecular programming is the powerful idea that we can program the very structure and behavior of matter at the molecular level.
Not only can we build nanostructures of nearly arbitrary shape and functional properties, but we can design molecules whose interactions encode computation.
Perhaps the earliest known example was able to compute solutions to NP-complete problems, in particular the Hamiltonian path problem~\cite{adleman1994molecular}.
Since then, implementations of models such as the Tile Assembly Model (TAM) and finite Chemical Reaction Networks (CRNs) have arisen~\cite{winfree1998algorithmic,soloveichik2010dna}, leading to the possibility of general Turing-universal computation~\cite{soloveichik2008computation}.

Versatile as these schemes are, for the most part our approaches to molecular computation are wasteful.
For one, they tend to be one-shot~\cite{adleman1994molecular,qian2011efficient,zhang2011dynamic}.
Computations are set up in a state far from equilibrium, with free energy dispersed throughout the various species (typically complexes of DNA).
Computation then typically corresponds to the process of equilibration of this system~\cite{simmel2019principles}.
While this certainly provides a high driving force for the experiment, it has a number of drawbacks.
In the case of many DNA-strand displacement systems, which typically consist of signal strands and gate complexes, many of the components cannot be reused; rather they are transformed into a plethora of waste complexes~\cite{zhang2011dynamic}.
As the variety of waste complexes is large, it is non-trivial to selectively remove them and replace them with fresh gate complexes.
Consequently, the runtime of computation is generally constrained by the initial concentrations of species.
Moreover, the prospect of a long-running computational system that responds to changing inputs over time---a feature that may be desirable in a smart therapeutic, for example---becomes impractical at best.
Nevertheless, some work on renewable or long-running dynamic DNA-based systems is beginning to emerge~\cite{del2022dissipative,eshra2017dna}.
As for Tile Assembly systems, free tile monomers are typically locked into a final assembly (except for transiently at the growing edge)~\cite{evans2017physical}.
Not only does this eventually starve future computation, but often these consumed tiles serve no ongoing functional purpose save that of storing the history of computation.
Signal-passing Tile Assembly models~\cite{padilla2014asynchronous} are capable of modifying the state of tiles after incorporation, such as to remove them, but these are irreversible changes.

Is this waste---of free energy, monomers, and special complexes---unavoidable?
Turning to life, the maestra of molecular machines, we see that this is not the case.
Biochemical systems routinely recycle monomers---building up and breaking down macromolecules from their constituent components.
Their molecular machines (enzymes), which we loosely identify with the gate complexes mentioned earlier, are in general fully reusable---acting catalytically.
Lastly, free energy is not distributed casually across a large number of components.
Rather, a select few species are designated as free energy currency.
The prototypical example is that of ATP/ADP.
This pair of chemical species can be converted between one another, so by maintaining a ratio of the two that is far from equilibrium, the hydrolysis of ATP into ADP can be used to store free energy.
The missing ingredient is then to build molecular machines that couple a desired reaction to this hydrolysis, so the free energy can be supplied to other reactions.
In this way, living systems need only inject fresh free energy into these few subsystems in order to continually sustain the operation of all other processes.

From this we draw one main conclusion for the design of more effective molecular systems.
Free energy should be separated out from the molecular machines it powers.
As a consequence, the molecular machines will generally be catalytic, returning to their original state after performing their function.
This is best implemented by exploiting reversible dynamics, which is already a characteristic of the microscopic realm.
In reversible dynamics, the previous state of a system (and indeed, its entire history) is uniquely determined by the current state.
More strongly, we have time reversal symmetry: the laws of physics are the same both forwards and backwards.
An interesting corollary of these reversible dynamics is that the system's evolution (at least at the local level) can be reversed by inverting the free energy supply.
Additionally, the speed of evolution can be controlled by increasing or decreasing the amount of free energy stored.
Of course, these are not unknown ideas to the molecular programming community.
One of the earliest proposals for a molecular computer, arguably predating the field, is Bennett's enzymatic Turing machine~\cite{bennett1982thermodynamics}.
Bennett's design consisted of a polymeric tape, bespoke enzymes performing reversible computational steps, and a pool of free energy currency and structural monomers.
While not experimentally realized, other enzymatic approaches to molecular computation have been, such as the PEN toolbox~\cite{montagne2011programming} and PER~\cite{kishi2018programmable}.
As designing custom enzymes remains highly non-trivial, however, these approaches reuse naturally occurring enzymes such as DNA polymerase.
Non-enzymatic approaches include the proposed DNA polymer stack machines of Qian et al.~\cite{qian2011efficient} and reversible surface CRNs of Brailovskaya et al.~\cite{brailovskaya2019reversible}.

We present a novel abstract model of molecular programming, Reversible Bond Logic (RBL).
RBL systems consist of `atoms', which can be combined by `bonds' into arbitrarily complex `molecules'.
By programming the energy landscape of the atom-bond configurations, reversible paths through configuration space can be carved.
This will prove sufficient to implement a variety of systems, including catalytic molecular machines and a common free energy currency.
Moreover, as RBL atoms are manipulated in a reversible fashion, they can be freely reused.
Indeed, in RBL complex and diverse macromolecular structures can be built up and broken down as required.
The informational content of these structures can then be exploited to build modular and compositional computational entities, allowing for rich computational primitives such as looping, recursion, and subroutines.

It is important to note that reversibility has profound implications on the very nature of computation, and so will affect how are our programs are written.
Conventional approaches to computation make liberal use of non-invertible computational primitives, such as overwriting variables or (freely) merging branches of control flow.
The most immediate consequence is that it is not generally possible to determine the previous state of a computer from the current state, as there are often many possible consistent histories.
While beyond the scope of this paper, there is a deep and rich connection between irreversibility in computation and the thermodynamics of information~\cite{szilard1929entropieverminderung,landauer1961irreversibility}.
As we seek reversible dynamics, we will leverage the principles of reversible computing~\cite{bennett1982thermodynamics} and programming~\cite{lutz1986janus,yokoyama2008principles}\cite[pp.~11--23]{earley2021performance}.
Namely, we will avoid loss of information and take care to distinguish branches of control flow when they are merged.

The paper begins with a definition of RBL.
Next, RBL is introduced more concretely with the implementation of a walker powered by an external fuel supply.
Then, one possible scheme for computation is presented; a small selection of structural-manipulation primitives are introduced, and used to implement conditional branching, looping, and subroutines.
Additionally, the computation is coupled to the same fuel supply introduced earlier, in order to drive the otherwise-unbiased reversible computation forward.
The paper concludes with a discussion of the advantages and limitations of RBL.
Elaboration of the computational primitives introduced and the design decisions behind them is presented in the technical appendix.

\section{Definition}

\begin{figure}
  \begin{subfigure}{0.55\linewidth}
    \small
    \setlength{\mathindent}{0pt}
    \begin{align*}
      \mathcal S &= (\{X,Y\}, \{P,Q\}, \{\text{solid}, \text{dashed}\}, \mathbb R_{\ge0}\cup\{\infty\}) \\
      P &= \{\text{solid},\text{dashed}\} \qquad Q = \{\text{solid}\} \\
      X &= (\{p,q\}, (p:P, q:Q), (p:\text{in}, q:\text{in}), e_X) \\
      Y &= (\{r,s\}, (r:P, s:P), (r:\text{out}, s:\text{in}), e_Y)
    \end{align*}\vspace{-1.5\baselineskip}
    \caption{}
    \label{fig:defex:formal}
  \end{subfigure}%
  \begin{subfigure}{0.25\linewidth}
    \fig{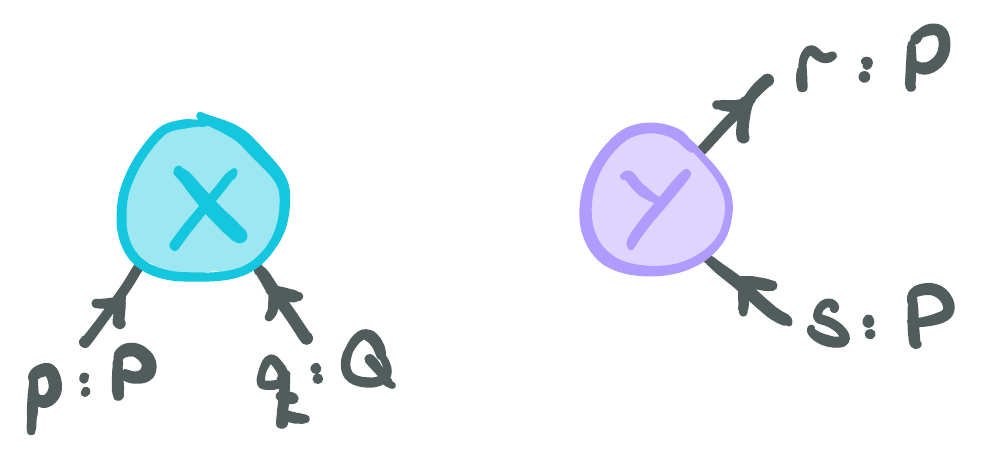}
    \caption{}
    \label{fig:defex:atom}
  \end{subfigure}%
  \begin{subfigure}{0.2\linewidth}
    \fig{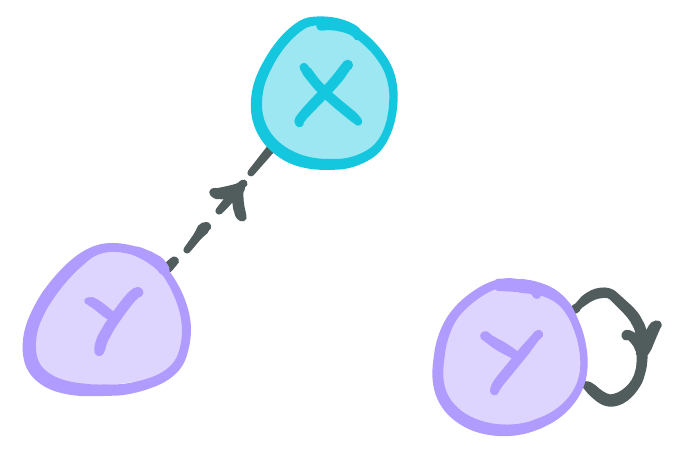}
    \caption{}
    \label{fig:defex:cfg}
  \end{subfigure}
  \caption{%
    An example RBL scheme, $\mathcal{S}$.
    \textbf{(\subref{fig:defex:formal})}
    The formal definition of $\mathcal S$, which comprises two atom types, $X$ and $Y$, two port types, $P$ and $Q$, and two bond colors, solid and dashed.
    Energies may be any non-negative real or infinite.
    Each atom has two ports, and the energy configurations ($e_X$, $e_Y$) are left unspecified.
    \textbf{(\subref{fig:defex:atom})}
    The RBL atoms defined in diagrammatic form.
    The label $p:P$ indicates that the port has label $p$ and is of type $P$.
    \textbf{(\subref{fig:defex:cfg})}
    An example system configuration with one copy of $X$ and two of $Y$. As we keep the same port positions as in (\subref{fig:defex:atom}), we omit port labels for brevity.
    }
  \label{fig:defex}
\end{figure}

Informally, an RBL scheme consists of a set of atoms.
These atoms are decorated with ports of certain types, and like ports can form bonds between one another---whether between ports on two different atoms or on the same atom (a self-loop).
Ports additionally come in `oriented' pairs: `in' and `out'; a bond is then formed \emph{directionally} from an out-port to an in-port.
A given port type may support bonds of multiple different `colors' (not necessarily literal colors), which can be useful for signal passing.
The decoration of ports on an atom is geometry-free: there is no intrinsic ordering or positioning of ports, and they can be considered freely labile.
There is some resemblance of RBL to Thermodynamic Binding Networks (TBNs)~\cite{doty2017thermodynamic}, which will be discussed in \Cref{sec:concl}.

Key to RBL is the logic that dictates permissible bond formation and transitions.
Each possible configuration of an atom---namely, whether each of its ports is bonded and, if so, with what color bond---is assigned an energy.
As expected, higher energies are thermodynamically less likely with `$E=\infty$' representing an impossible configuration.
In fact, in practice configurations are usually only assigned energies of $E=0$ or $E=\infty$, i.e.\ possible or impossible, although other finite positive energies may be employed for short-lived transitional states.
A system configuration---the configuration of multiple atoms---may transition to another configuration if and only if:
  (1) the configurations differ by a single bond (equivalently, by two port states),
  (2) the configuration energies are finite.
That is, only one bond transition (whether formation, color change, or breakage) may occur at a given time.
The kinetics may depend on the energy and associated entropy changes, e.g.\ due to a change in particle number, but we leave a detailed treatment of this to future work.

Designing an RBL scheme then amounts to carefully picking the energy landscapes of each atom so as to carve out a guided path(s) through configuration space.
In general, we seek to restrict the system so that at any one time there are precisely two possible transitions: one to the previous state, and one to the next state.
Sometimes, however, it may be desirable to provide parallel paths through configuration space when the order of operations is unimportant.

Although we will introduce RBL schemes diagramatically, we define RBL formally for completeness.
An RBL scheme is specified by a tuple $(A,P,B,E)$ where $A$ is an indexed set of atoms, $P$ an indexed set of ports, $B$ a set of bond colors, and $E$ a set of possible energies.
The indices of $A$ and $P$ are used to uniquely label each type of atom and port.
A port $p$ consists of a subset of $B$, i.e.\ $P_p\subseteq B$, corresponding to the bond colors it can form.
An atom of type $\alpha$ consists of a tuple $(J,\vec p,\vec o,e)$, where $J$ is an index set enumerating the ports of the atom, $\vec p$ is a vector of port types indexed by $J$, $\vec o$ is a vector of port orientations indexed by $J$, and $e$ is an energy function.
An atom's energy function maps each possible configuration to its energy, i.e. $e:\prod_{\ell\in J}(\{\varepsilon\}\cup P_{p_\ell}) \to E$ where $\varepsilon$ corresponds to an unbound port.
An example scheme is shown in \Cref{fig:defex}.

A system configuration is a directed graph.
Each atom of type $\alpha$, with atom-tuple $(J, \vec p, \vec o, e)$, has an associated sub-graph.
This sub-graph consists of an `atom' node labeled $\alpha$ and a set of `port' nodes:
  for each $j\in J$ we add an edge labeled $j$ from the atom-node to the port-node, which is labeled $(p_j, o_j)$.
A configuration consists of a (disjoint) union of atom-graphs, where there may be any number $\in\mathbb N$ of copies of each type of atom-graph.
Bonds correspond to edges from a node $(p,\mathrm{out})$ to a node $(p,\mathrm{in})$ for some port type $p$, and are labeled with some color in $P_p$.
For each atom-sub-graph, we can compute its energy using the energy function $e$.
The system configuration is valid if each of the energies of its constituent atoms is finite, and the total energy of the system is given by the sum of these.
Two system configurations are \emph{adjacent} if they are both valid and they differ by a single port-port edge;
  this difference can be the presence/absence of such an edge, or a change in label.
A system may transition to any adjacent configuration.

\section{Walkers and Fuel}

A walker is a molecular machine that walks along a molecular track.
In nature these are often referred to as motor proteins, and serve a vital role in cells performing such tasks as transporting cargo or contracting muscle cells.
Molecular programmers have designed many instances of walkers.
Some walkers have no requirement of directional movement and perform a random walk through their domain, such as the `cargo-sorting robot' of Thubagere et al.~\cite{thubagere2017cargo}.
Other walkers, such as that of Yin et al.~\cite{yin2008programming}, achieve uni-directional movement along a one dimensional track by using high-free energy fuel to permanently block previously-visited track footholds.
This can be considered analogous to a Brownian ratchet.
Such `burnt bridge' approaches to walking have the disadvantage that it is not possible to send multiple walkers down the same track, and the track needs to be rebuilt in order to walk along it again.
Motor proteins do not modify their track.
Instead, processive movement is conferred by supply of an external free energy currency.
We will implement such a walker within RBL, starting with an unbiased walker that has no directional preference and then adding both polarity and an external fuel supply that couples directional movement to the free energy supply.

\subsection{Unbiased Walker}

\begin{figure}
  \begin{subfigure}{0.55\linewidth}
    \centering
    \begin{tabular}{ccccc}
      atoms &\quad& $E=0$ &\quad& $E=\infty$\\
      \fig{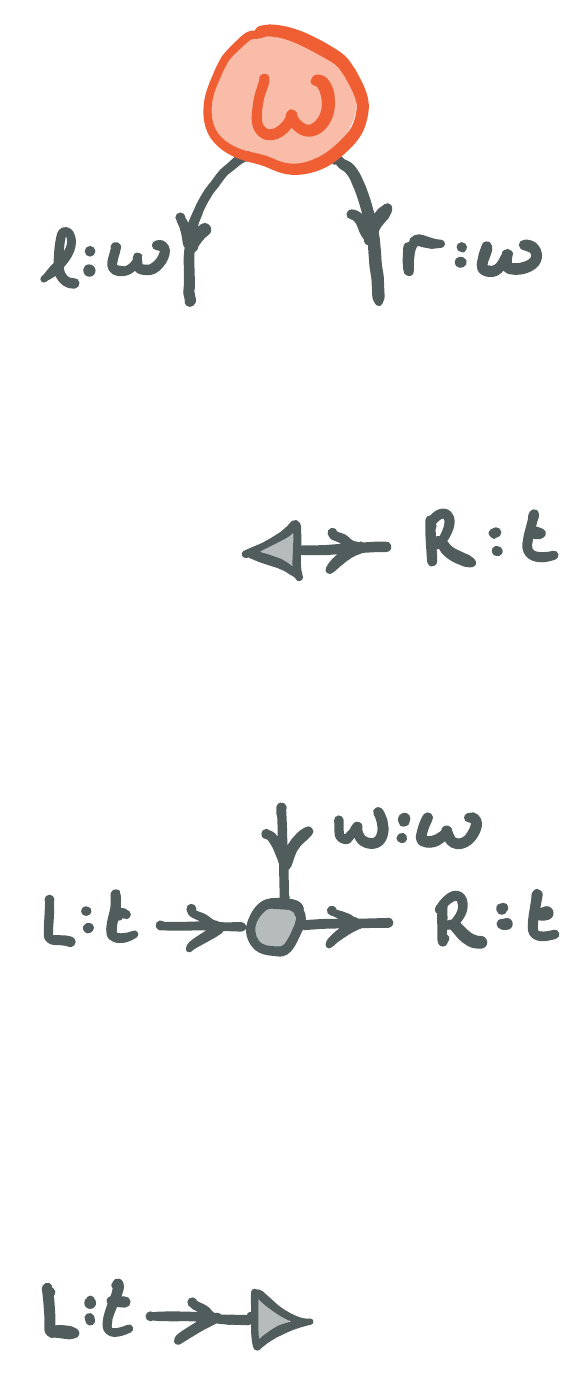} &&
      \fig{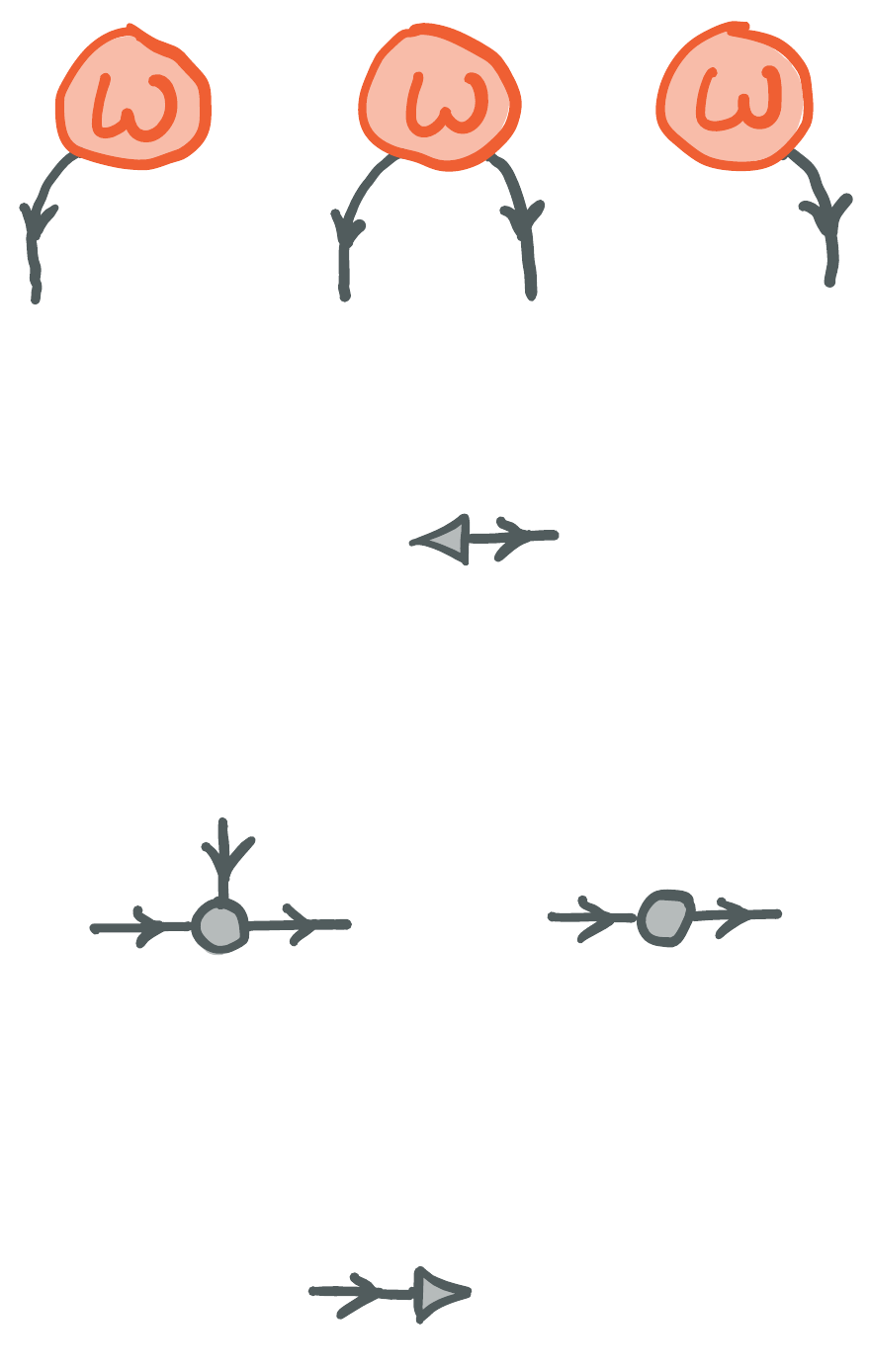} &&
      \fig{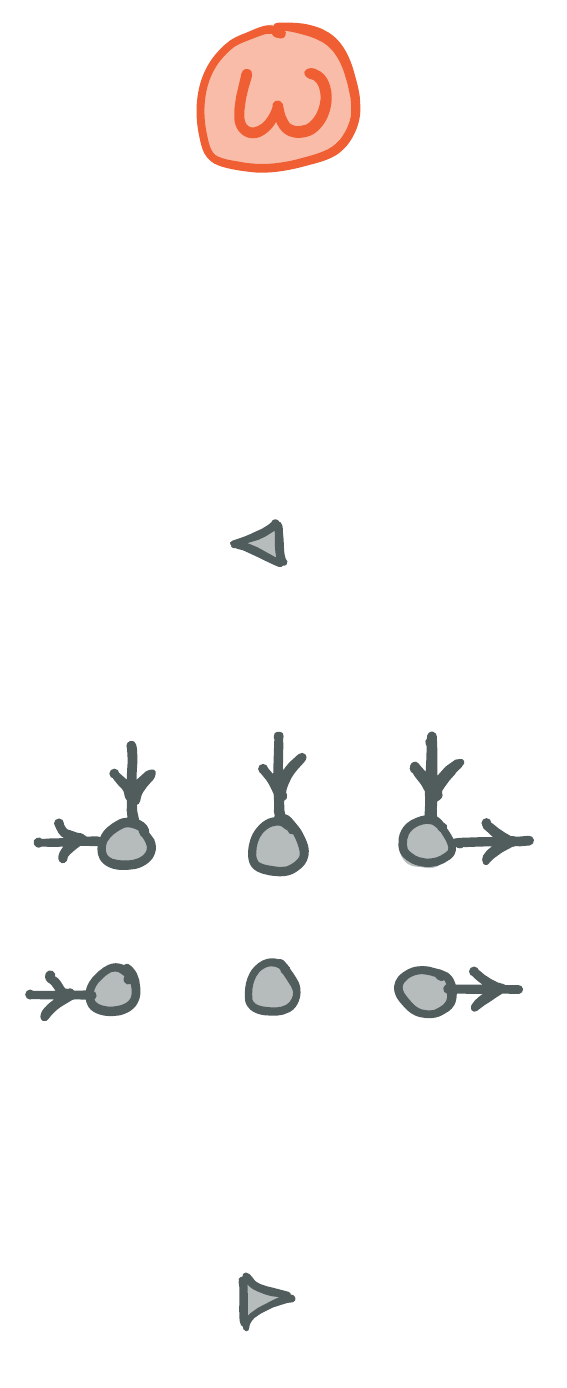}
    \end{tabular}
    \caption{The RBL energy configuration of an unbiased walker.}
    \label{fig:walku-config}
  \end{subfigure}\hfill%
  \begin{subfigure}{0.35\linewidth}
    \centering
    \fig{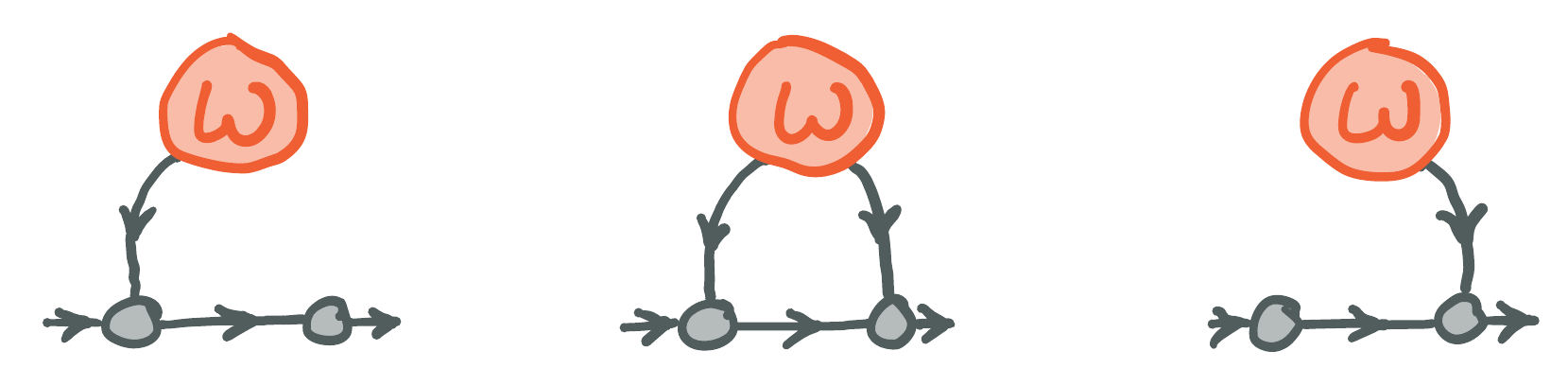}
    \caption{The dynamics of walking.}
    \label{fig:walku-ex}
  \end{subfigure}
  \caption{%
    An RBL implementation of an unbiased walker.
    \textbf{(\subref{fig:walku-config})}
    The system consists of four types of atom. 
    \atom{W} represents the walker;
      it has two feet, given by the left ($\ell$) and right ($r$) ports, which are both of type $w$ and can form bonds of a single color, solid.
    The other atoms represent the track, and comprise a left-cap, track monomers, and a right-cap.
    The track monomers have a $w$ port of type $w$, allowing the feet of the walker to bind to them.
    The other ports, $L$ and $R$ of type $t$ (monochromatic solid), allow the track monomers to polymerize; if capped on the ends they will form a linear track, else a circular loop.
    The energy configuration is defined by assigning each possible state to $E=0$ (allowed) or $E=\infty$ (impossible).
    The energy configuration of \atom{W} allows a walker to bind to one or two track monomers, but it can never dissociate from the track as the unbound state is impossible.
    Meanwhile, the energy configuration of the track enforces that it must be complete and cannot fall apart.
    Optionally, the `transition' state where \atom{W} binds to two track monomers can be assigned a higher energy, e.g.\ $E=1$.
    \textbf{(\subref{fig:walku-ex})}
    The energy configuration leads to the walking dynamics shown.
    The walker is free to step to the left or right, forming a transient state with two feet on the track.
    Then either foot can dissociate, possibly leading to net movement along the track.
    }
  \label{fig:walku}
\end{figure}

An unbiased walker is relatively easy to implement, and an appropriate RBL scheme is shown in \Cref{fig:walku}.
The walker itself is given by a single RBL atom, \atom{W}, and it walks along a track using its two `feet' ports to make bonds as needed.
All that is required for correct operation is for \atom{W} to be able to bind to the track with one or both feet, but not none.
In this way, it can make `progress' by putting down a second foot and lifting the first foot.
As the dynamics are reversible and history-free, there is of course no guarantee that it won't just lift the second foot and make no progress, nor is there any guarantee that it will make net progress over time.
Indeed, the statistics of movement are described by an unbiased random walk with zero mean displacement after $t$ steps, and $\sim\mathcal O(\sqrt{t})$ variance.

A brief commentary on the geometry of the walker atom \atom{W} is warranted.
If the ports of \atom{W} had fixed position, then there might be a risk of the walker not moving beyond the initial pair of track monomers.
Free rotation about the bonds or sufficient flexibility would fix this, however recall from the definition of RBL that the ports have no intrinsic position and are freely labile.
As such, the ports are free to `move' as needed.
A further consequence of being geometry-free is that a walker may `skip' an arbitrary number of positions along the track; there are a number of ways to address this, but we defer discussion of this to future work which will introduce a geometric model for RBL.

Also of note is the importance of the directionality of the bonds.
Without directionality, the $w$ ports of the track monomers would be able to bind to each other, as would the $\ell$ and $r$ feet of the walker \atom{W}.
Therefore directionality enforces a certain complementarity relation, which is not unfamiliar in the world of DNA nanostructures.

\subsection{Fuel}

In order to build a powered directional walker, we will need an external fuel supply.
Inspired by biochemistry, we choose to store our free energy in the population of two related species.
For simplicity, we design (\Cref{fig:fuel}) a bistable atom \atom{G} with two distinct states, $\atom G_+$ and $\atom G_-$, that can be interconverted by an external signal at the coupling port $f$.
Consequently, the equilibrium state corresponds to $[\atom G_+]=[\atom G_-]$ and free energy may be stored by increasing the ratio $[\atom G_+]:[\atom G_-]$ above 1.
Specifically, the amount of free energy is given simply by $kT\log([\atom G_+]/[\atom G_-])$.
For another system $X$ to usefully access this free energy, it needs to couple carefully to the fuel species.
The initial state of the system $X_{\text{init.}}$ must couple to the $\atom G_+$ state, and the final state of the system $X_{\text{fin.}}$ must couple to the $\atom G_-$ state.
That is, we need to implement the `reaction' $X_{\text{init.}}+\atom G_+ \rightleftharpoons X_{\text{fin.}}+\atom G_-$.
This is the purpose of the colorings available on the $f$ port, and will be shown more concretely in the next subsection.

\begin{figure}
  \centering
  \begin{tabular}{clc}
    atom &\qquad& $E=0$ \\
    \fig{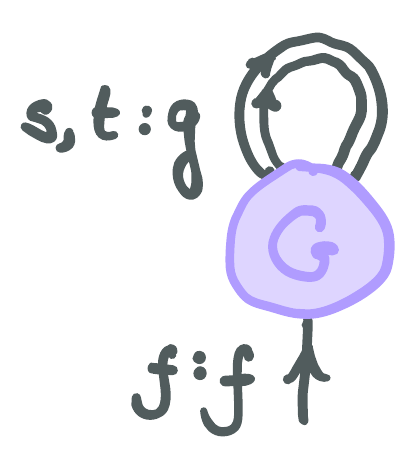} &&
    \fig{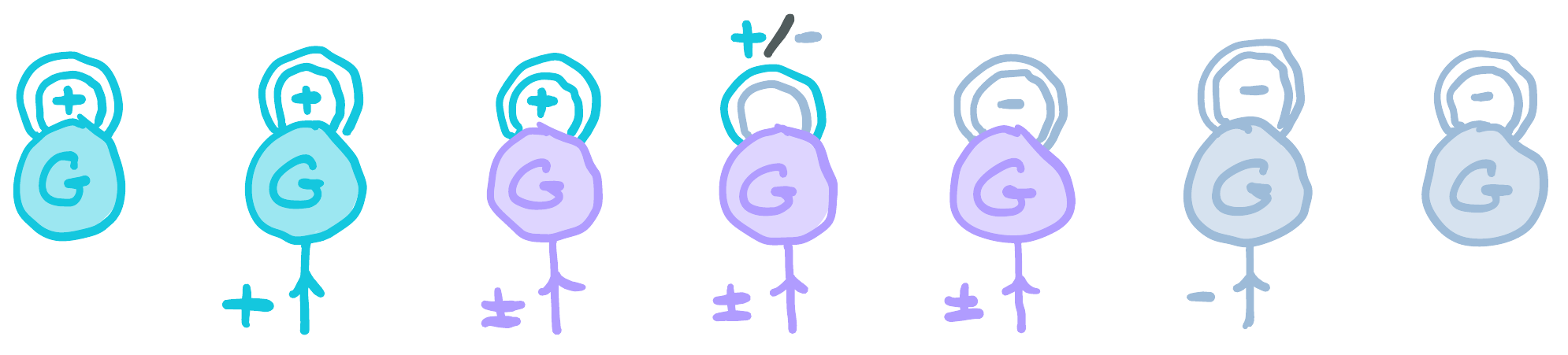}
  \end{tabular}
  \caption{%
    An RBL implementation of a bistable fuel species, \atom{G}.
    The atom has five ports, $s_1/s_2$, $t_1/t_2$, and $f$.
    The ports $s_1/s_2$ form a self-loop, and so we refer to them together as $s$; similarly for the self-loop $t$.
    The loops are each of type $g$, which has two colors: $+$ (bright blue) and $-$ (gray-blue).
    The `coupling' port $f$ is of type $f$ and has three colors: $+$, $\pm$ (lilac), and $-$.
    The two stable states of \atom{G} correspond to both loops being $+$ ($\atom{G}_+$) or both being $-$ ($\atom{G}_-$).
    These can be interconverted via the coupling port, $f$, passing through a transitional state $\atom G_\pm$ where $s$ is colored $+$ and $t$ is colored $-$.
    The $+$ and $-$ colorings allow another RBL atom to distinguish $\atom{G}_+$ from $\atom{G}_-$, whilst the $\pm$ coloring allows interconversion.
    If there is more $\atom{G}_+$ than $\atom{G}_-$ then there will be a negative (favorable) free energy change $\Delta G$ associated with the reaction $\atom{G}_+ \to \atom{G}_-$.
    Note that we only show the possible states ($E=0$); all other configurations may be presumed impossible ($E=\infty$).
    }
  \label{fig:fuel}
\end{figure}

Critical to the design is the double self-loop on \atom{G}.
A single self-loop would be able to autonomously transition between bond colors due to the rules of RBL, but a double self-loop allows us to impose an (infinite) energy barrier between the $\atom G_+$ and $\atom G_-$ states.
This barrier is lowered by the action of the coupling port, $f$, which thus acts as a catalyst for the interconversion of states.

\subsection{Biased Walker}
\label{sec:walk}

\begin{figure}
  \begin{minipage}[b]{0.25\linewidth}
    \begin{subfigure}{\linewidth}
      \centering
      \fig{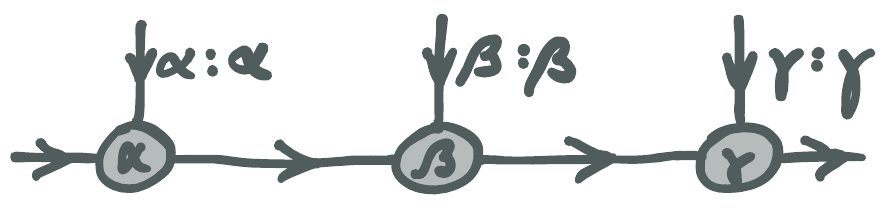}
      \caption{The track.}
      \label{fig:walk-track}
    \end{subfigure}%
    \\[\baselineskip]
    \begin{subfigure}{\linewidth}
      \centering
      \fig{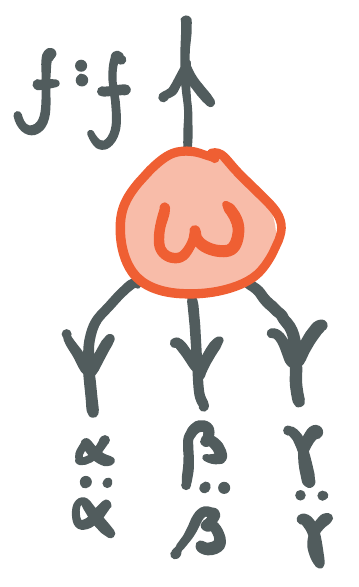}
      \caption{The walker, \atom{W}.}
      \label{fig:walk-mono}
    \end{subfigure}
  \end{minipage}
  \hfil\hfil
  \begin{subfigure}[b]{0.6\linewidth}
    \centering
    \fig{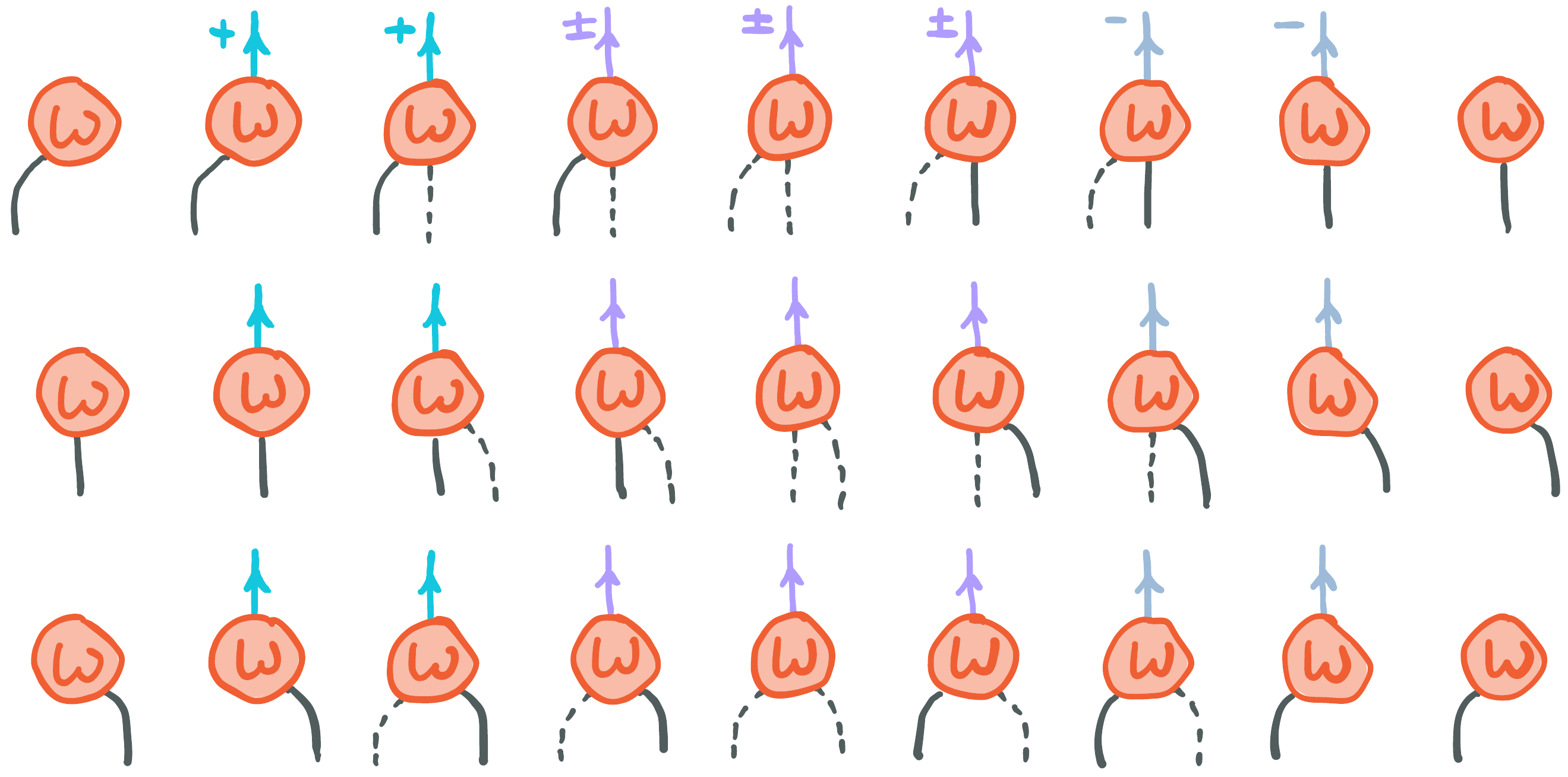}
    \caption{The permissible energy configurations of \atom{W} ($E=0$).}
    \label{fig:walk-config}
  \end{subfigure}
  \hfil%
  \\[\baselineskip]
  \begin{subfigure}{\linewidth}
    \centering
    \fig{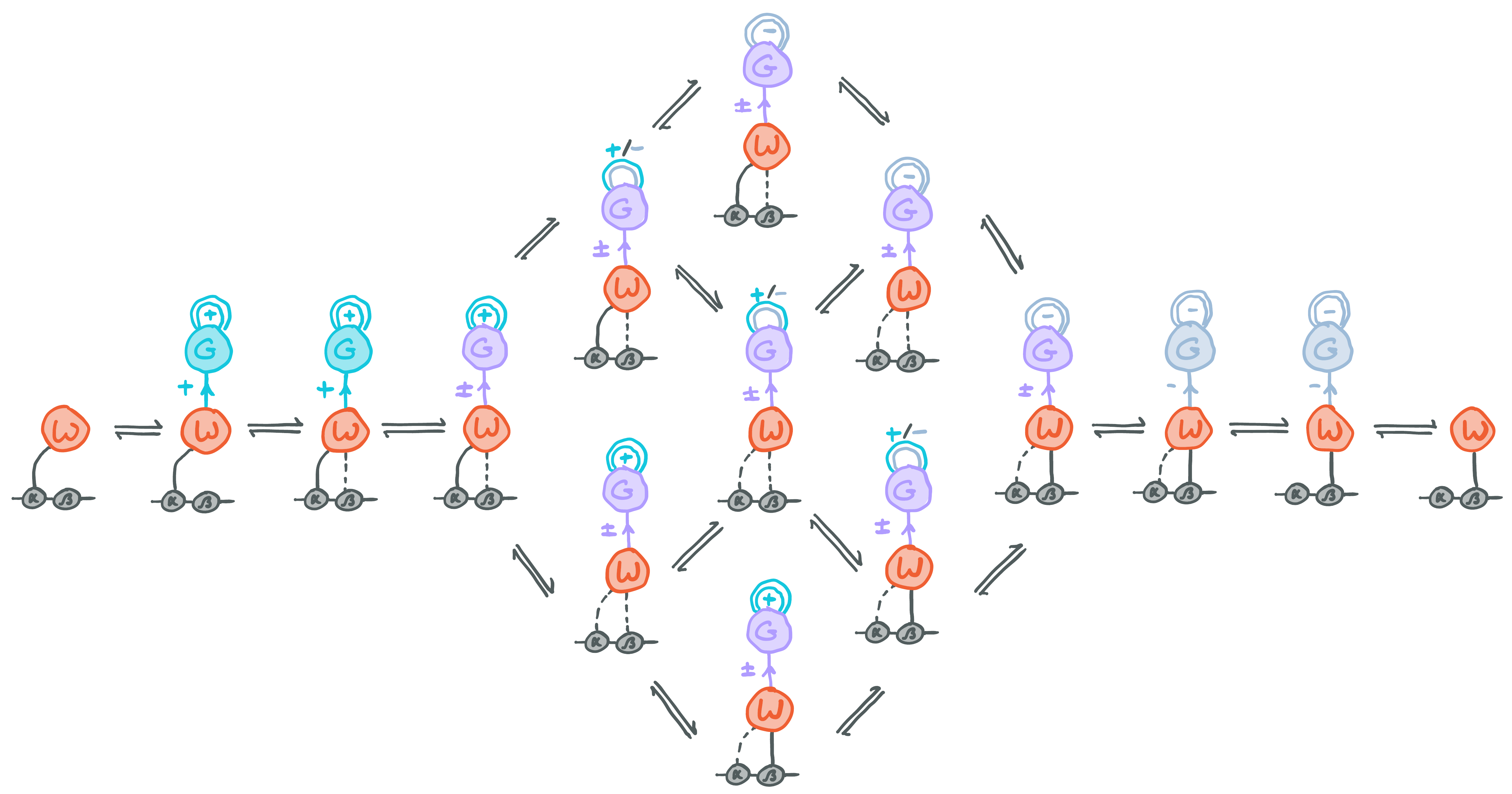}
    \caption{The transition diagram for an $\alpha\to\beta$ step, i.e.\ $\atom W_\alpha + \atom G_+ \rightleftharpoons \atom W_\beta  + \atom G_-$.}
    \label{fig:walk-ex}
  \end{subfigure}
  \caption{%
    An RBL implementation of a biased walker.
    \textbf{(\subref{fig:walk-track})}
    The track consists of a polymer of $\alpha$, $\beta$, and $\gamma$ track monomers.
    Their implementation is nearly identical to that of the track in \Cref{fig:walku-config}, except that:
      (1) the $L$/$R$ ports enforce the $\cdots\alpha\beta\gamma\alpha\cdots$ sequence, by $R_\alpha$/$L_\beta$, $R_\beta$/$L_\gamma$, and $R_\gamma$/$L_\alpha$ each having a distinct port type;
      and (2) each monomer's foothold port has a distinct type ($\alpha$, $\beta$, or $\gamma$).
    Another difference is that the foothold ports have two bond colors: `solid' and `dashed', with the latter corresponding to a transitional state.
    The track can again be linear or cyclic.
    \textbf{(\subref{fig:walk-mono})}
    The walker atom \atom{W} has three feet which specifically bind each type of track monomer.
    It also has a fuel coupling port $f$.
    \textbf{(\subref{fig:walk-config})}
    The set of permissible energy configurations of \atom{W} implements each of the `reactions' in \Cref{rxns:walk}, one per row.
    Considering the top row, we start with \atom{W} bound to an $\alpha$ track monomer.
    If it then binds to $\atom G_+$, it can begin to attempt an $\alpha\to\beta$ step
      (looking to the end of the third row, if it instead bound $\atom G_-$ then it would begin to attempt an $\alpha\to\gamma$ backstep).
    From here, it tentatively (dashed bond color) places a foot on the $\beta$ monomer.
    Then it changes the color of $f$ to $\pm$ (lilac), the transitional state.
    While the fuel is interconverting between $\atom G_+$ and $\atom G_-$, it changes its $\alpha$ binding to dashed and its $\beta$ binding to solid in two steps.
    Then, it attempts to change the color of $f$ to $-$ (gray-blue), which would indicate the successful consumption of fuel ($\atom G_+\to\atom G_-$).
    It can then dissociate its $\alpha$ foot followed by the spent fuel $G_-$.
    While we have described this step as if \atom{W} has `intent', it is important to remember that this is only for narrative benefit; the system is really performing a random walk through configuration space, with bias provided by the free energy stored in the concentration of the fuel.
    \textbf{(\subref{fig:walk-ex})}
    This $\alpha\to\beta$ step process is expounded in the transition diagram, which shows which system configurations can be reached from which other system configurations.
    This makes clear that the consumption of fuel and transition of foothold states occur in parallel.
    }
  \label{fig:walk}
\end{figure}

To construct a walker with a preferred walking direction, we must introduce polarity into both the track and the walker.
In nature, a single monomer type (e.g.\ actin) is sufficient for this by taking advantage of its spatial substructure.
In RBL, however, such spatial substructure is non-existent.
Instead we employ an alternating sequence of monomers.
Two monomers, i.e.\ a sequence $\cdots\alpha\beta\alpha\beta\cdots$, would be insufficient to distinguish forward from backward movement; the minimum sequence to introduce polarity consists of three monomers, i.e.\ $\cdots\alpha\beta\gamma\alpha\beta\gamma\cdots$.
With such a polarized track, the walker can then be modified such that, from an initial position on monomer $\alpha$, the forward step takes it to $\beta$ and the backward step to $\gamma$.
Similar behavior applies to the other track positions.
Of course, without a fuel supply forward and backward steps are identical and so movement is still non-directional.
To complete the implementation, we couple these steps to consumption of fuel as follows:
\begin{align}
  \atom W_\alpha + \atom G_+ &\rightleftharpoons \atom W_\beta  + \atom G_- &
  \atom W_\beta  + \atom G_+ &\rightleftharpoons \atom W_\gamma + \atom G_- &
  \atom W_\gamma + \atom G_+ &\rightleftharpoons \atom W_\alpha + \atom G_-
  \label{rxns:walk}
\end{align}
There are many possible designs, but we choose a three-footed walker with a fuel coupling port $f$.
Each foot is specific to a particular monomer type, i.e.\ $\alpha$, $\beta$, or $\gamma$, and the dynamics of `walking' resemble a wheel rolling along the track.
This design is elaborated in \Cref{fig:walk}.
The main challenge in designing such an RBL system is to ensure that correct system configurations cannot jump to invalid configurations.
For example, if \atom{W} did not maintain a foothold on both $\alpha$ and $\beta$ during an $\alpha\to\beta$ step, then it would be possible to `jump' into either of the other step processes.
To assist in this, a computational suite (to be released in the future) for designing and testing RBL schemes was developed.
To finish, we prove (\Cref{thm:walk}, \Cref{app:prf}) that the designed scheme has the desired properties:

\begin{restatable}{theorem}{biasRandomWalk}
  \label{thm:walk}
  The dynamics of the biased walker, in the long-run, are that of a biased random walk.
\end{restatable}

\section{Data and Computation}

The advantages of structured data and programming abstractions are well known to users of high-level programming languages.
In this section, we will develop a set of conventions and motifs for representing and computing with structured data in RBL.
We will use these to implement three example `programs':
  (1) logical negation of a Boolean value, with which we will introduce \textbf{conditional branching};
  (2) addition of natural numbers (using a Peano representation), with which we will introduce \textbf{looping};
  (3) squaring of natural numbers, using addition as a \textbf{subroutine}.
A recursive implementation of addition is left as an exercise for the reader.
High level schemata to motivate these implementations are presented in \Cref{fig:rxn-hi-lvl}.

\begin{figure}[p]
  \begin{subfigure}{\linewidth}
    \centering
    \fig{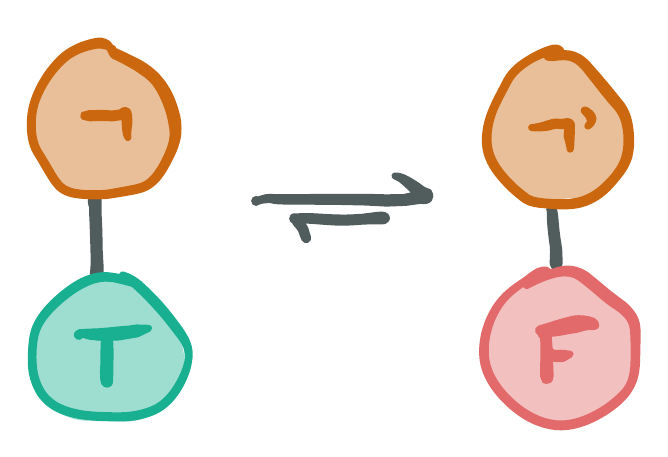}\hskip 1cm
    \fig{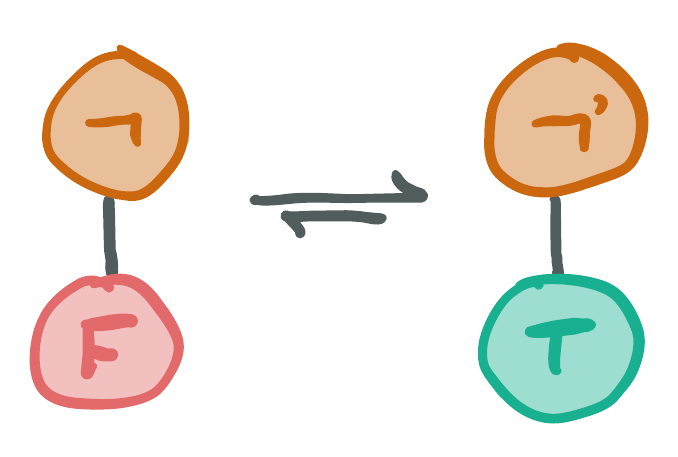}
    \caption{%
      \textbf{Logical negation.}
      A Boolean value, \atom T or \atom F, is tagged with an atom $(\neg)$ which serves as a `reference' to the logical negation `function'.
      Then, two `reactions' are implemented recognizing each of the possible inputs.
      The return value is tagged with $(\neg')$ to represent the result of the logical negation.
      }
    \label{fig:rxn-not}
  \end{subfigure}
  \begin{subfigure}{\linewidth}
    \figSm{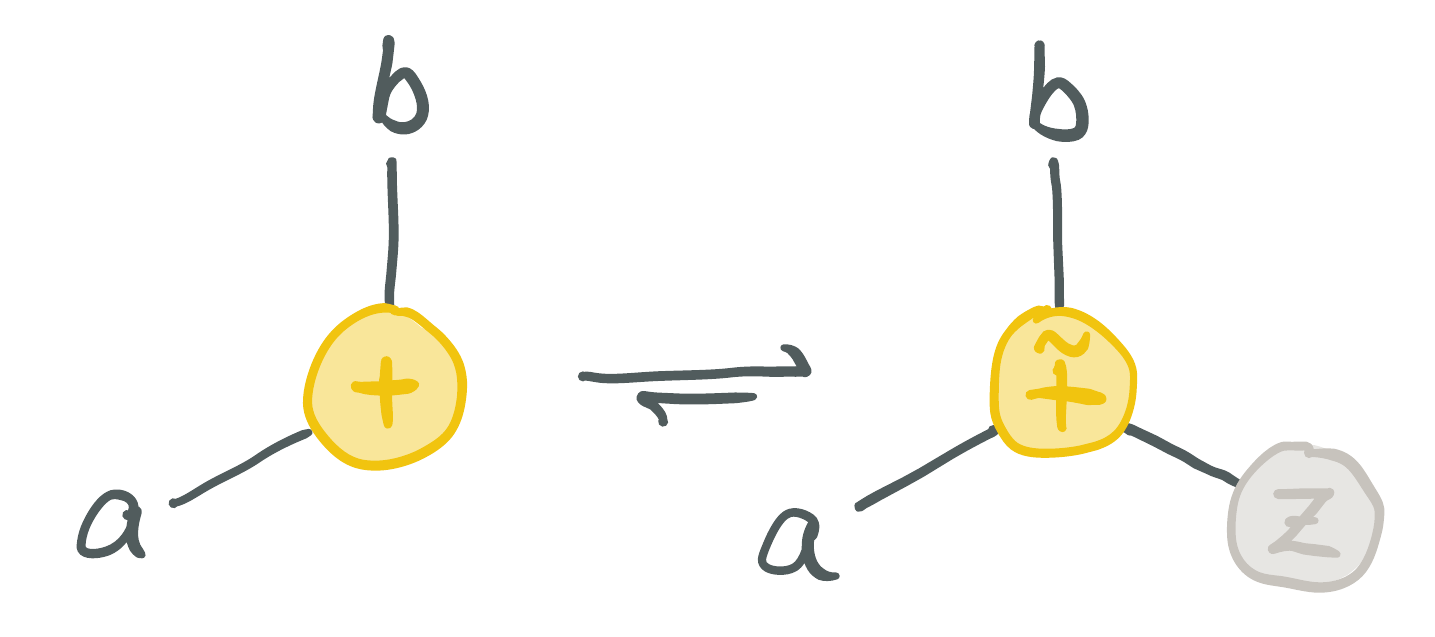}\hfill
    \figSm{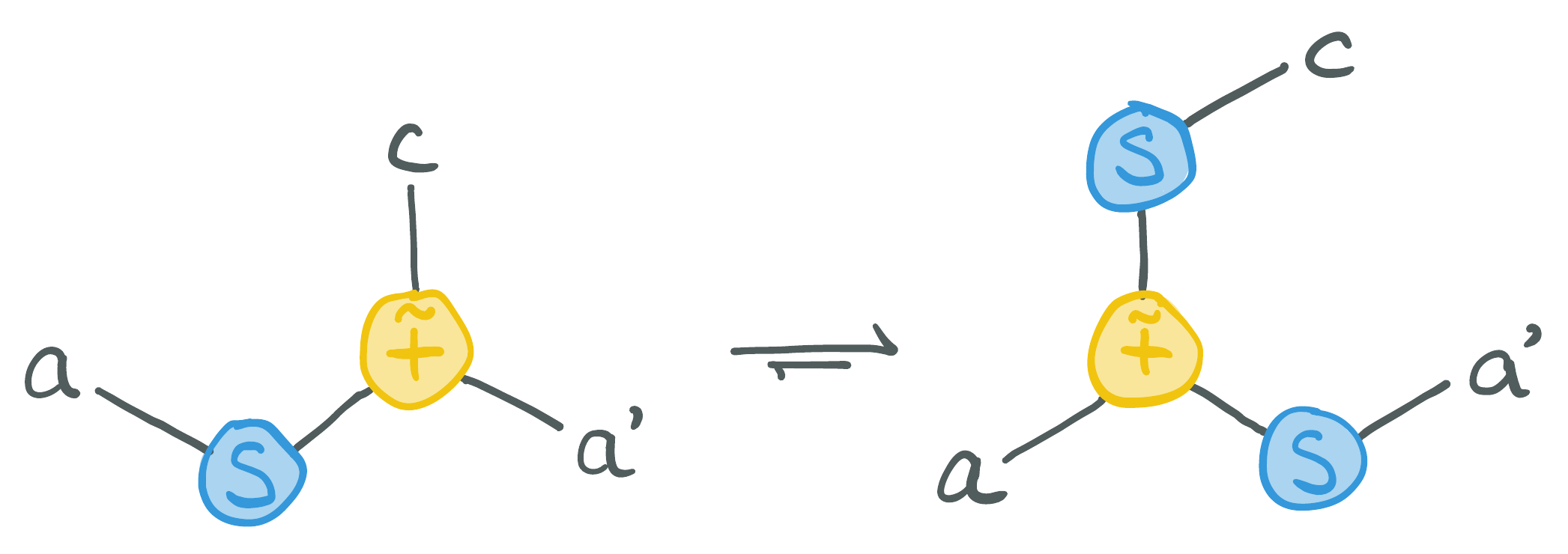}\hfill
    \figSm{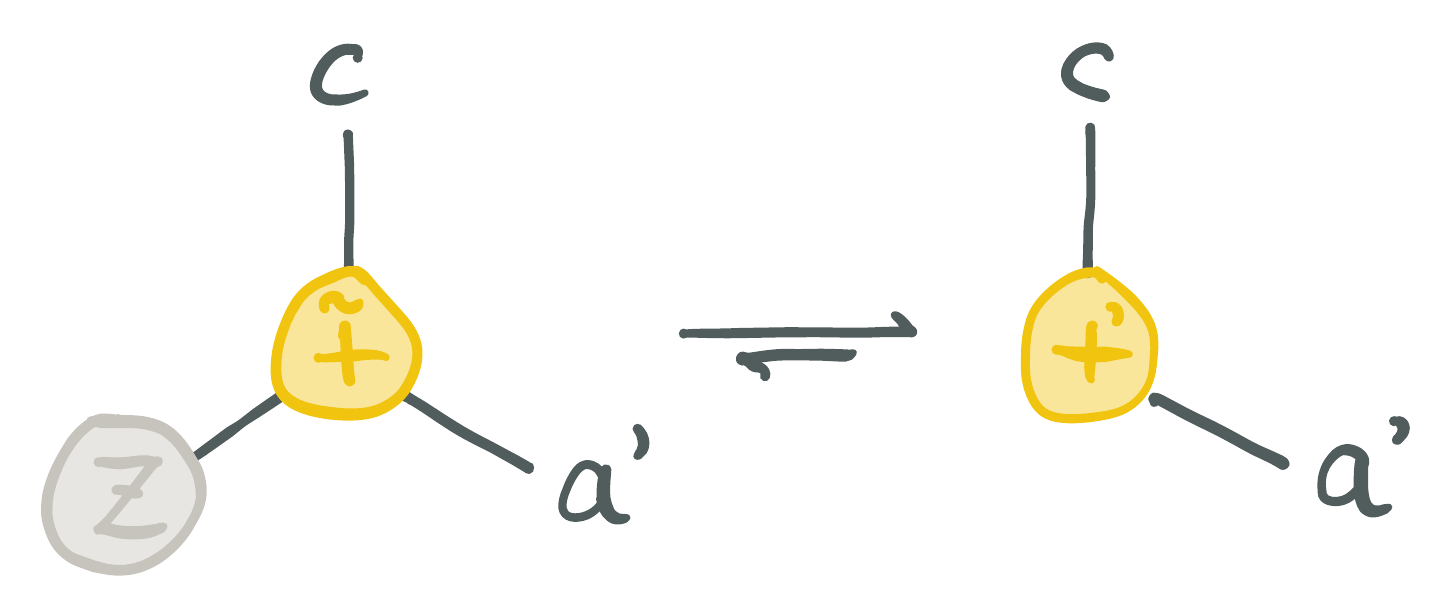}
    \caption{%
      \textbf{Addition.}
      As with logical negation, we use two special atoms to represent an addition to be performed, $(+)$, and the result of an addition, $(+')$.
      As the computation must be reversible, we retain one of the inputs: for example, after adding $a=3$ and $b=4$, we would return $c=7$ and $a'=3$.
      The inputs and outputs are represented as Peano numbers, i.e.\ polymers of the form \atom{S}\bond\atom{S}\bond$\cdots$\bond\atom{Z} where the number of \atom{S} atoms corresponds to the number.
      For example, $3\equiv\text{\atom{S}\bond\atom{S}\bond\atom{S}\bond\atom{Z}}$.
      The core of the computation is a loop (middle reaction, tagged by $(\tilde +)$), which iteratively decrements the bottom-left number and increments the top and bottom-right numbers.
      To see that this reaction functions as a loop, notice that the product can serve as a reactant to the same reaction while the loop condition holds.
      Note also how any possible molecule can participate in at most two reactions (corresponding to a forward and backward transition), giving a deterministic path from input to output.
      As the bottom-left number is originally $a$ and the top number is originally $b$, these become respectively $0$ and $c=a+b$.
      The bottom-right number starts at 0 and becomes $a'=a$.
      The left reaction is responsible for entering the loop, and the right reaction for exiting the loop.
      }
    \label{fig:rxn-add}
  \end{subfigure}
  \begin{subfigure}{\linewidth}
    \centering
    \figSm{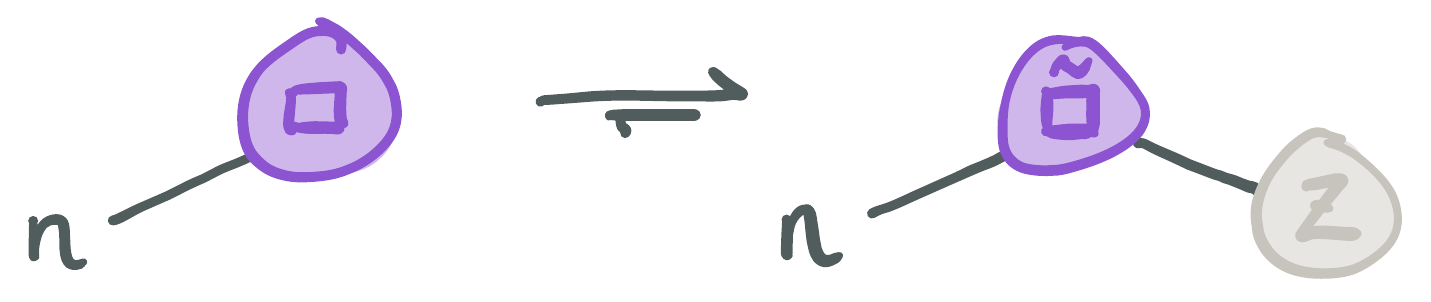}\hskip 3cm
    \figSm{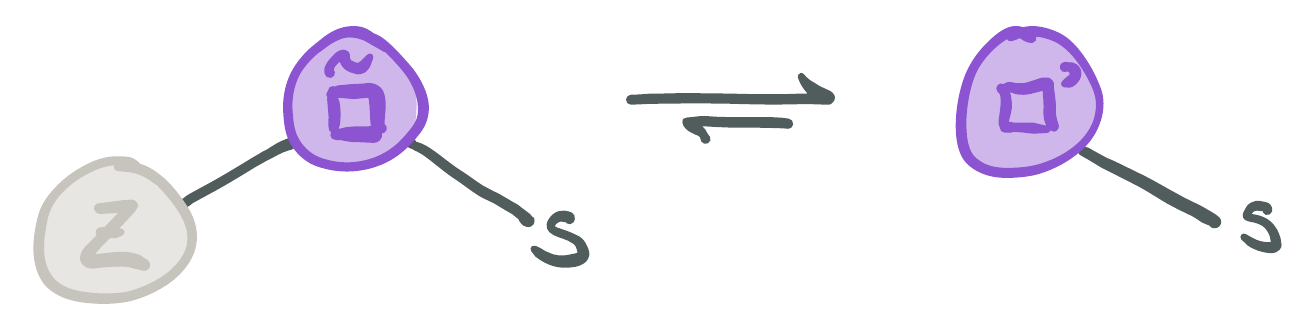}\\[1em]
    \figSm{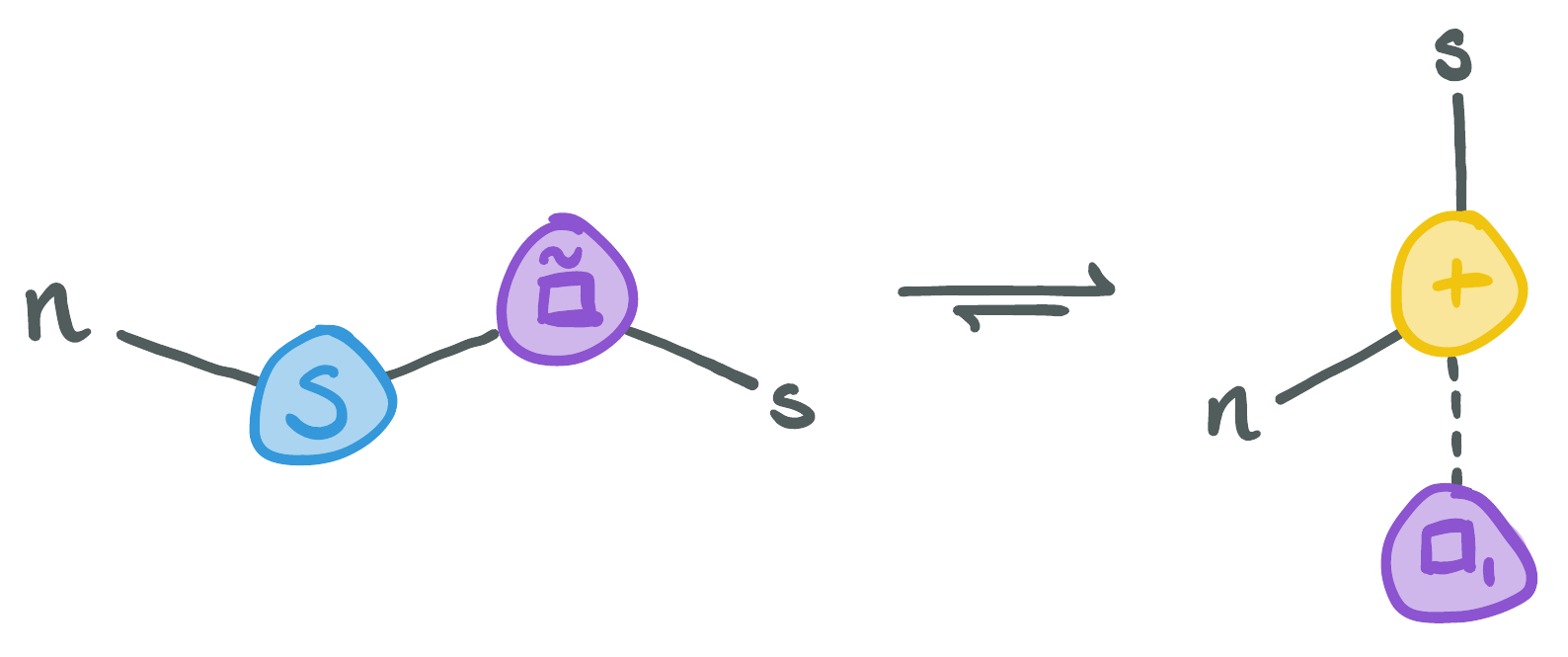}\hfill
    \figSm{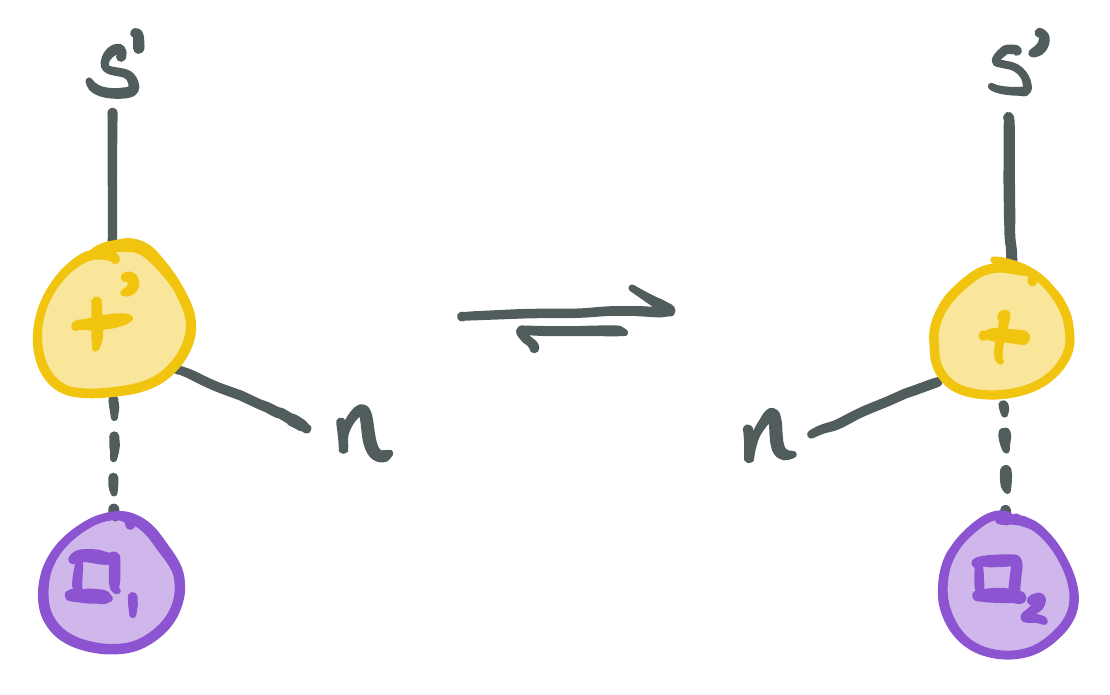}\hfill
    \figSm{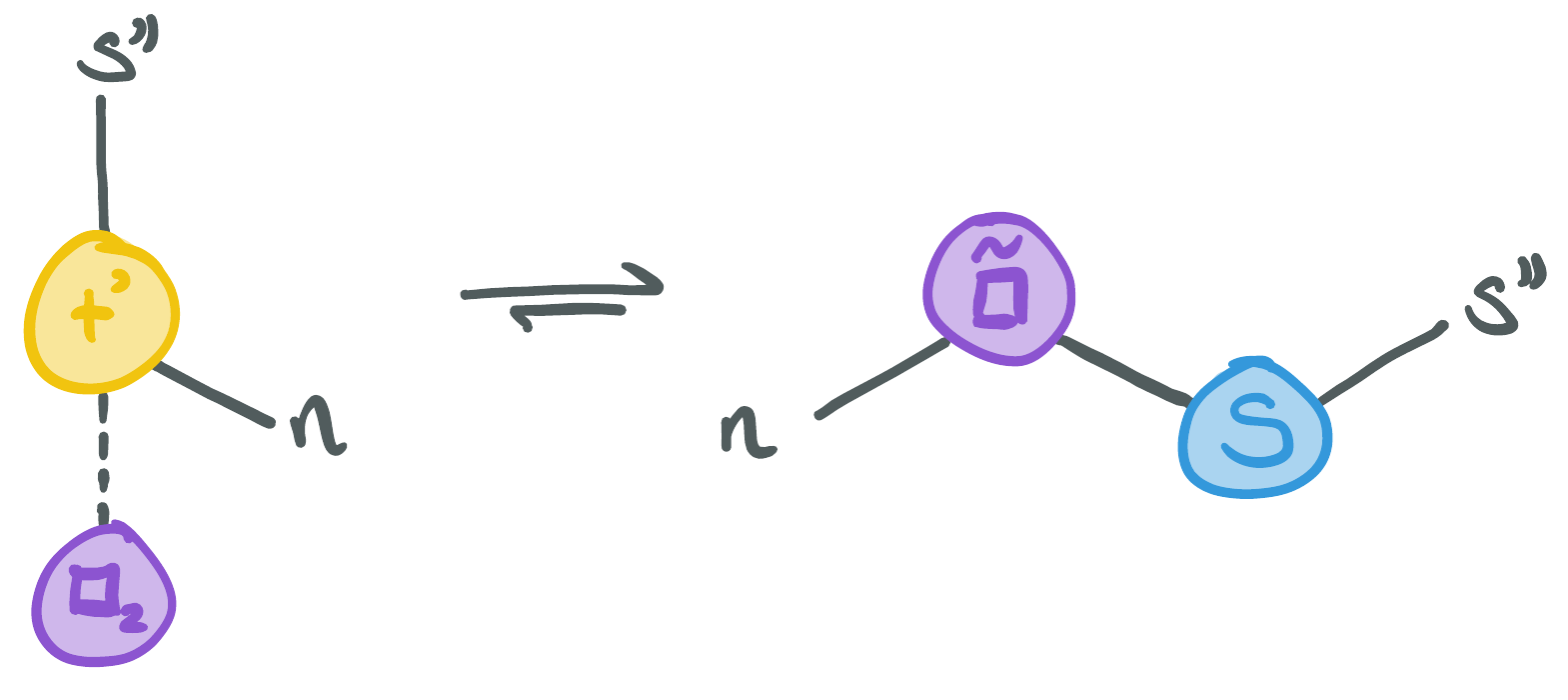}
    \caption{%
      \textbf{Squaring.}
      A number $n$ can be squared by using the identity $n^2=(2n-1)+(2n-3)+\cdots+3+1$:
        that is, $n^2$ is the sum of the first $n$ odd numbers.
      The implementation of squaring uses $(\square)$ and $(\square')$ as initial and final tags.
      It takes a Peano number $n$ as input and returns a Peano number $s=n^2$ as output.
      We implement squaring using a loop (tagged by $(\tilde\square)$) with two variables, $n$ and $s$.
      Initially, $n$ is the number to be squared and $s=0$.
      On each iteration, we decrement $n$, add this to $s$ twice, and then increment $s$; the net result is $s\mapsto s+(2n-1)$ and $n\mapsto n-1$.
      The addition is performed by using the program in \Cref{fig:rxn-add} as a subroutine.
      By starting with the largest odd number and ending at 1, we ensure that we consume the value $n$ (by reducing it to 0) and produce only $s=n^2$ as an output, even though addition always returns one of its inputs.
      }
    \label{fig:rxn-sq}
  \end{subfigure}
  \caption{%
    High level schematic representations of three simple programs.
    These are not themselves RBL systems, but will be used as a blueprint for the construction of equivalent RBL systems.}
  \label{fig:rxn-hi-lvl}
\end{figure}

\begin{figure}[p]
  \begin{subfigure}[t]{0.6\linewidth}
    \centering
    \fig{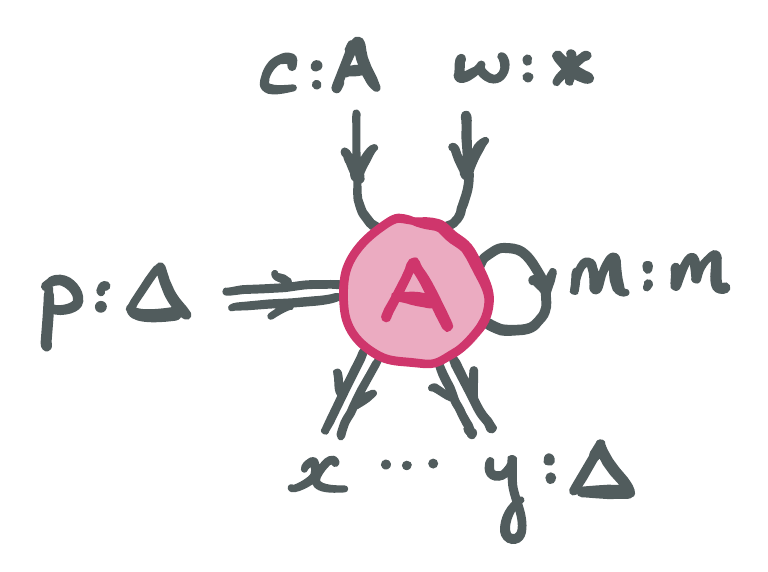}
    \qquad
    \fig{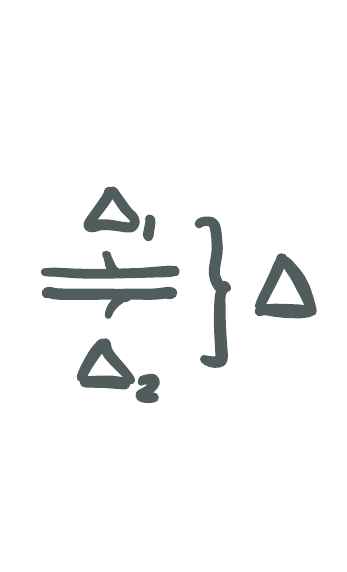}
    \caption{%
      The general form of a data-atom and data-port ($\Delta$).
      Its $c$ and $w$ ports correspond to the specific and wildcard control ports respectively, and the $m$ port pair corresponds to the `monomer' self-loop.
      The $p$ data-port binds the data-atoms parent, and $x\cdots y$ provide bindings to some number of children.}
    \label{fig:data-mono:data}
  \end{subfigure}%
  \hfill%
  \begin{subfigure}[t]{0.35\linewidth}
    \centering
    \fig{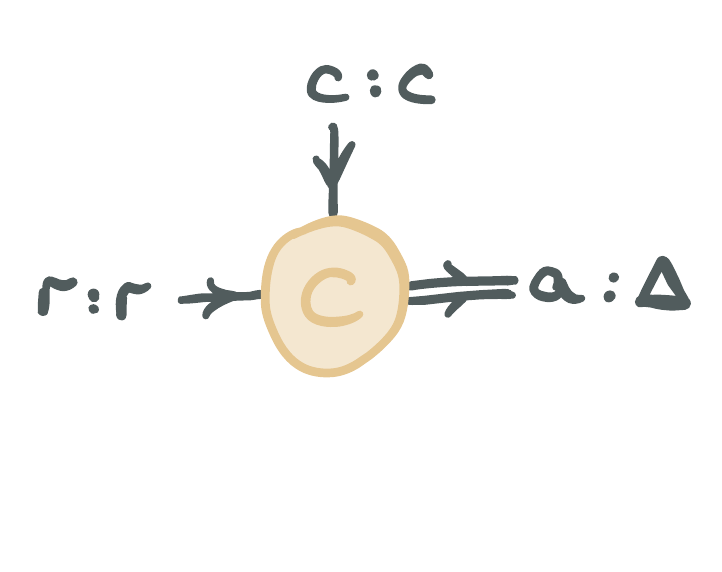}
    \caption{%
      The \atom{C} atom.
      The $a$ data-port binds some computational data, the $r$ port binds its root or origin, and the $c$ port may be bound by some compuzyme.}
    \label{fig:data-mono:C}
  \end{subfigure}
  \caption{%
    The RBL atoms for data-atoms and the \atom{C} atom.
    The energy configurations and dynamics are explored in \Cref{app:motif}.}
  \label{fig:data-mono}
\end{figure}

\subsection{Motifs}
\label{sec:motif}

Recall that we intend for the computational systems we design to be modular, compositional, and to support component reuse.
To realize these properties, we will require a common convention for the representation of data.
Moreover, as this data may be reused in different contexts (for example, Peano numbers may be added or multiplied), the data itself should generally be inert.
Computation will be performed by dedicated catalytic `machines', taking inspiration from biochemical systems.
We will call these `compuzymes' (computational enzymes).
Consider the schematic for addition (\Cref{fig:rxn-add}).
In principle, we will be able to construct RBL molecules corresponding to each of the three types of abstract molecule in the scheme, and we can also construct a compuzyme implementing each of the three reactions.
To demonstrate the power of RBL, however, we will construct a single compuzyme performing the entire addition (loop included).

Data will be constructed from `data-atoms', RBL atoms with a certain structure.
Consider Peano numbers, where a number is either zero (\atom Z) or the successor (\atom S) of another number.
In haskell, this may be represented as \texttt{data $\mathbb N$ = Z | S $\mathbb N$}.
We can therefore see that the \atom{S} atom has one `child' while the terminal \atom{Z} atom has no children.
For an example of a data-atom with more than one child, consider the nodes of a binary tree%
  \footnote{Recall that the definition of a binary tree in haskell is
            \texttt{data Tr $a$ = Lf | Nd (Tr $a$) $a$ (Tr $a$)}.}.
These data-atoms would have three children, two for its child nodes (or leaves), and one for its associated data.
As all data-atoms could be the child of another atom, each data-atom also has a parent.
Therefore, an RBL representation of a data-atom must have an in-port for a parent to bind, and an out-port for each of its children; we refer to these special ports as `data-ports'.
These ports will in fact be pairs of ports: as with the double self-loop on the fuel atom \atom{G}, a pair of bonds can be used to provide an energy barrier against unintentional bond-breakage.

In addition to parent and child data-ports, a data-atom needs a few other ports.
To manipulate data, we add two control ports.
These allow compuzymes to externally signal data-atoms to alter their configuration.
One control port is specific to the type of data-atom; for example, the control port of an \atom{S} atom is distinct from that of a \atom{Z} atom.
The other control port is a `wildcard' control port that is common to all data-atoms, and enables (limited) manipulation of variable data.
Lastly, we have a pair of ports that can form a self-loop.
This self-loop is present only on free monomers.
The reason why this is required is somewhat subtle, and is explained in \Cref{app:motif}.

The final key component of data is the \atom{C} atom.
Some data represent computations (either to be performed or that have completed), such as the molecules in \Cref{fig:rxn-hi-lvl}.
These data are capped by a \atom{C} atom, e.g.\ \atom{C}\bond$(\neg)$\bond\atom{T}.
The \atom{C} atom therefore identifies computational data.
It may seem extraneous, but it provides important indirection for subcomputations by allowing two compuzymes to interact with the same computational data simultaneously.
As such, \atom{C} has a `root' port ($r$) bound by the computation's origin, a `compuzyme' port ($c$) optionally bound by a compuzyme, and a data-port ($a$) binding the actual computational data.

\begin{figure}
  \begin{subfigure}{0.3\linewidth}
    \centering
    \renewcommand\thesubfigure{M}
    \fig{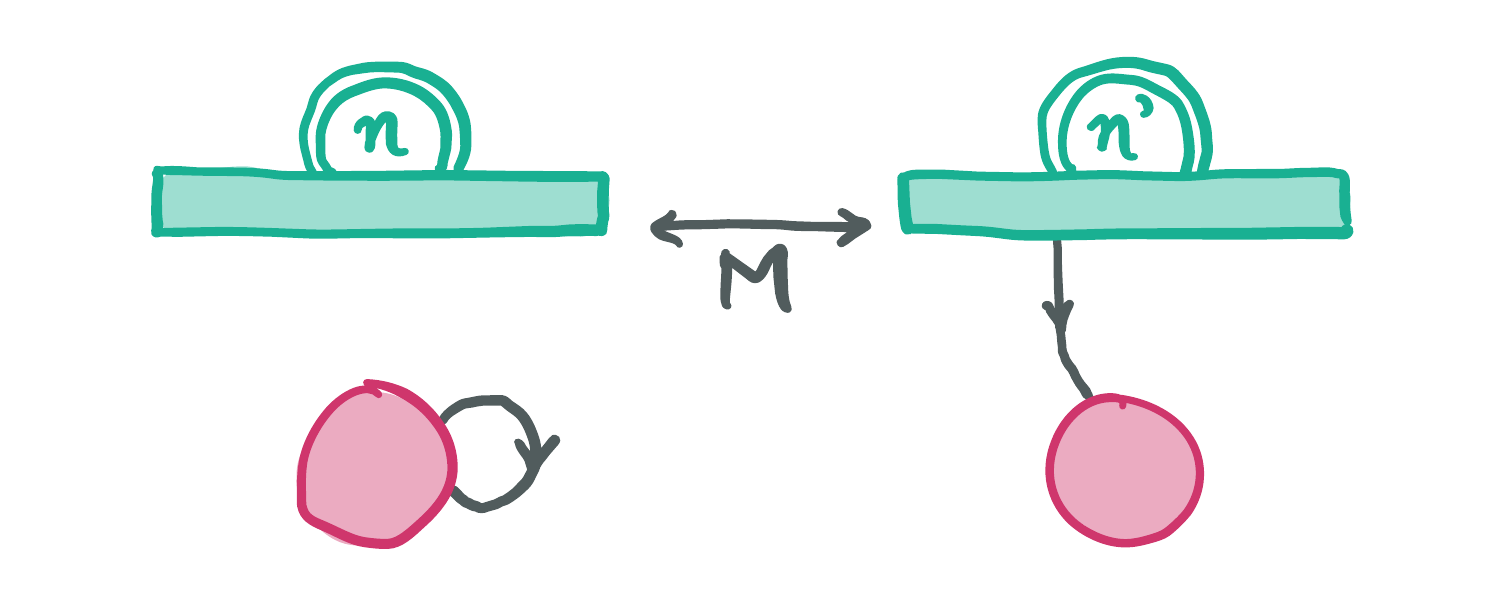}
    \subcaption{monomer binding}
    \label{fig:motif:mbind}
  \end{subfigure}%
  \hfil%
  \begin{subfigure}{0.3\linewidth}
    \centering
    \renewcommand\thesubfigure{X}
    \fig{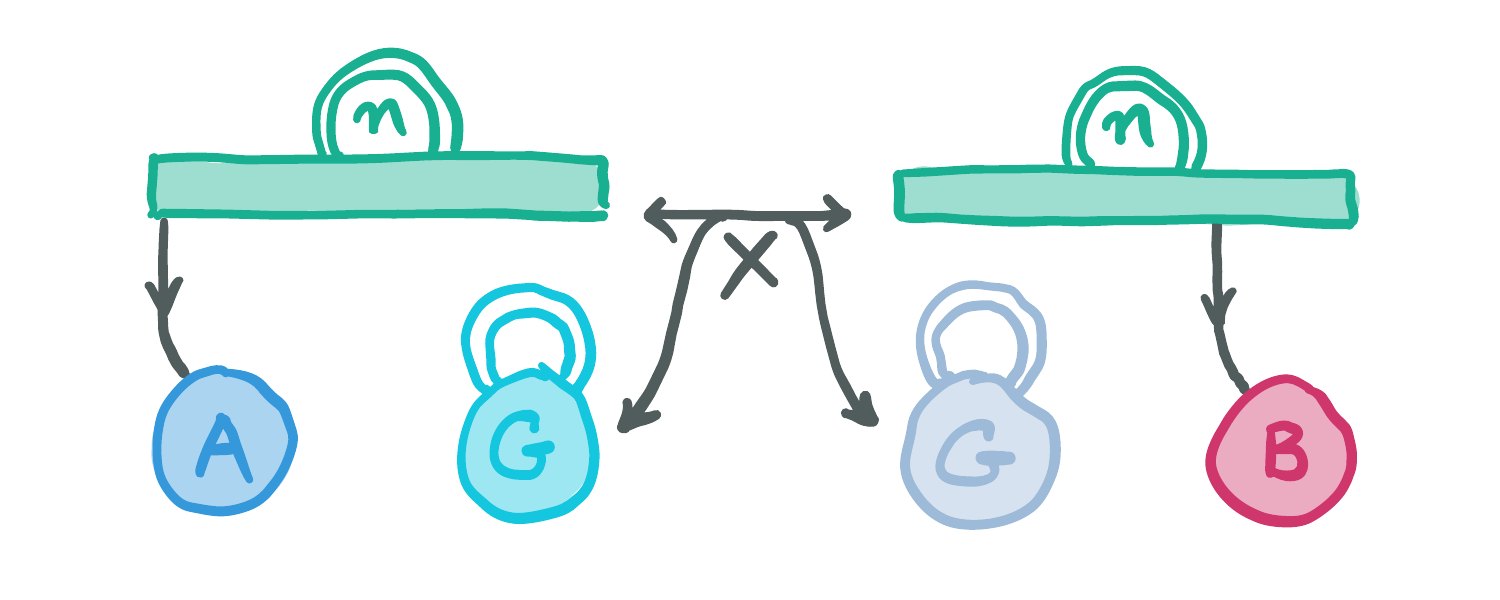}
    \subcaption{monomer exchange (fuelled)}
    \label{fig:motif:mono-xchg}
  \end{subfigure}%
  \hfil%
  \begin{subfigure}{0.3\linewidth}
    \centering
    \renewcommand\thesubfigure{C}
    \fig{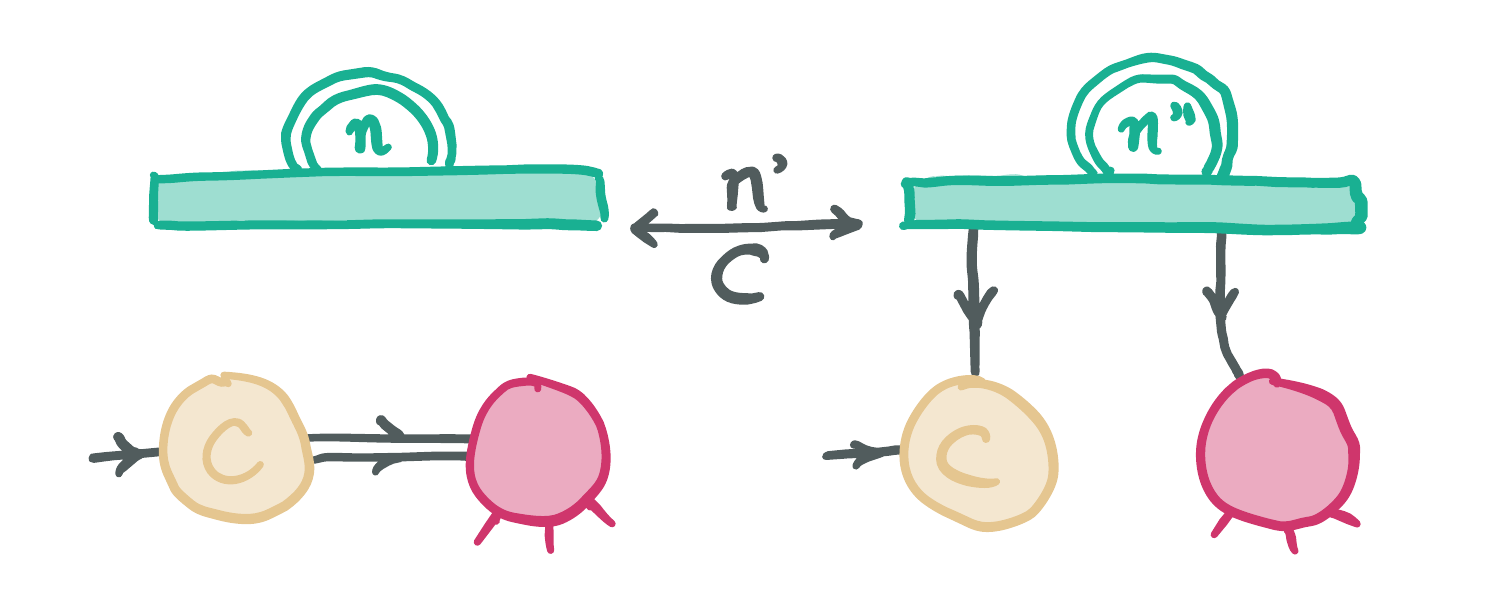}
    \subcaption{computational binding}
    \label{fig:motif:cbind}
  \end{subfigure}%
  \\[\baselineskip]
  \begin{subfigure}{0.3\linewidth}
    \centering
    \renewcommand\thesubfigure{D}
    \fig{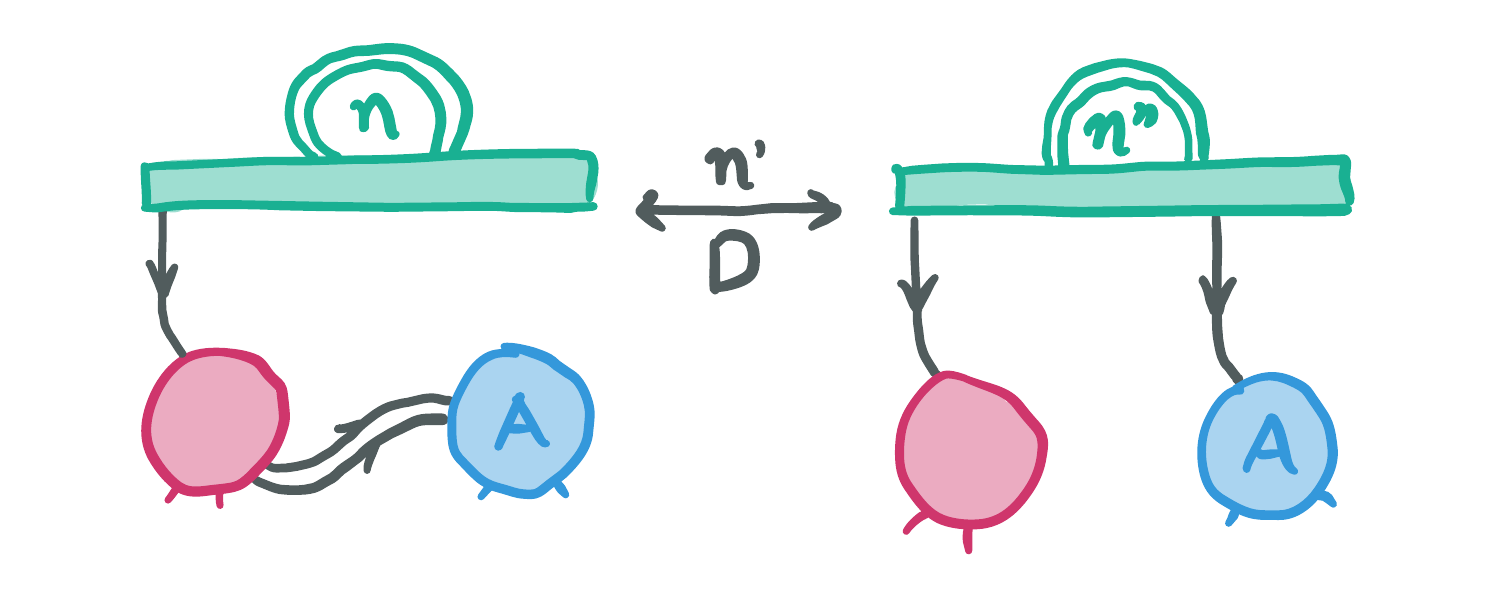}
    \subcaption{destructuring}
    \label{fig:motif:destruct}
  \end{subfigure}%
  \hfil%
  \begin{subfigure}{0.3\linewidth}
    \centering
    \renewcommand\thesubfigure{D$^\ast$}
    \fig{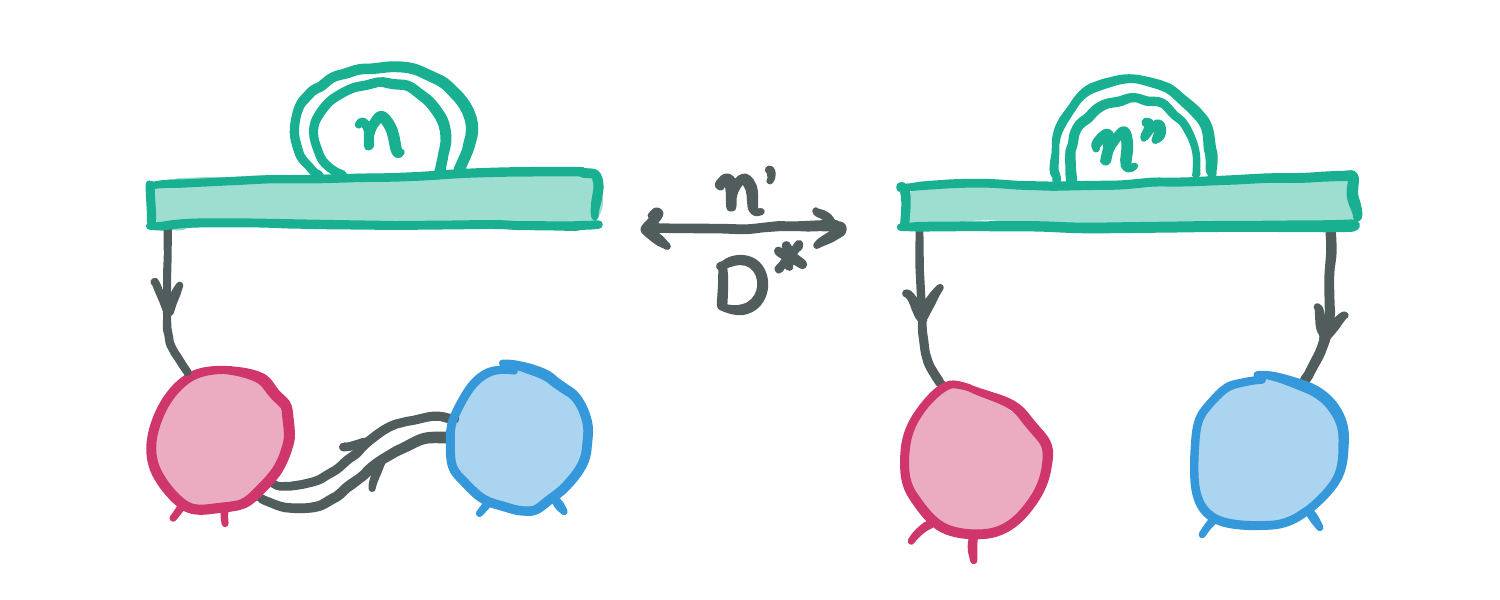}
    \subcaption{wildcard destructuring}
    \label{fig:motif:destruct-w}
  \end{subfigure}%
  \hfil%
  \begin{subfigure}{0.3\linewidth}
    \centering
    \renewcommand\thesubfigure{S}
    \fig{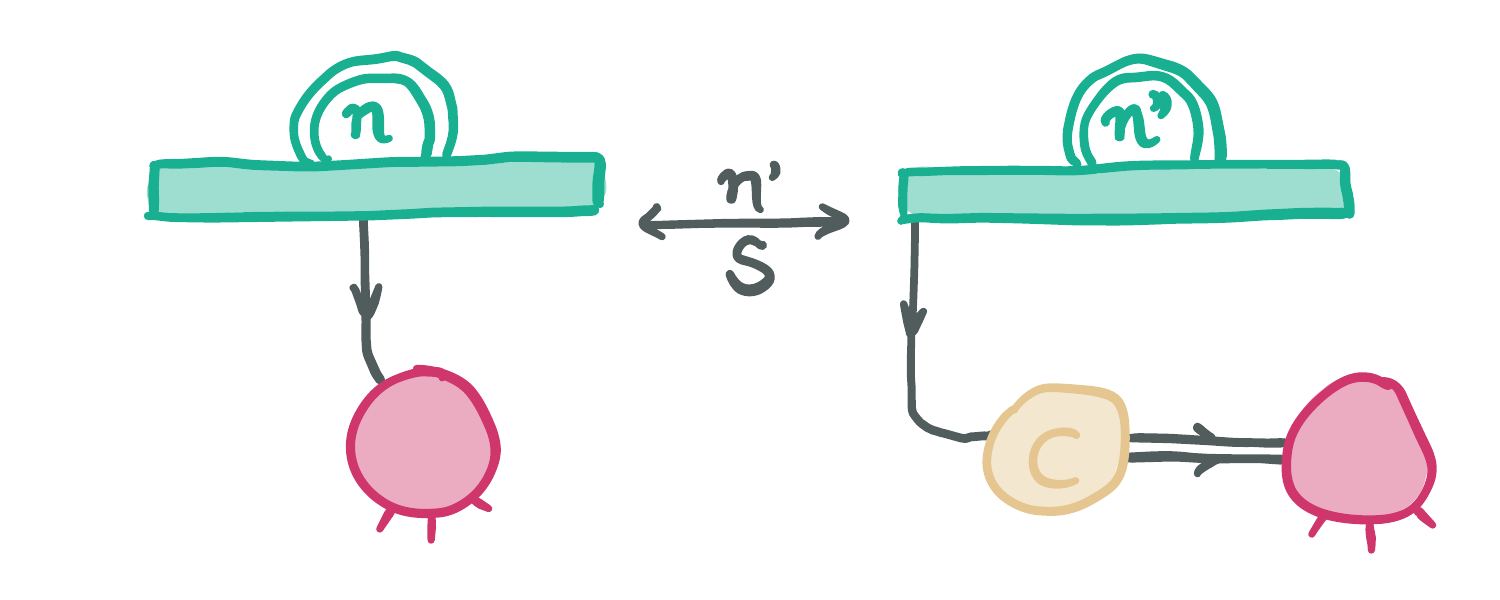}
    \subcaption{subcomputation initiation}
    \label{fig:motif:sub-init}
  \end{subfigure}%
  \caption{%
    The five core motifs underlying our computational convention.
    Only the net effect of the motif is shown; the detailed RBL implementations are given in \Cref{app:motif}.
    The green rectangle is a compuzyme, and the encircled value $n$ corresponds to the internal compuzyme state.
    In all motifs except X, there are one or two compuzyme state changes (from $n$ to $n'$, and possibly to $n''$).
    \textbf{(\subref{fig:motif:mbind})}
    The binding of a free monomer (pink data-atom).
    \textbf{(\subref{fig:motif:mono-xchg})}
    The exchange of a monomeric data-atom \atom{A} for \atom{B}.
    \textbf{(\subref{fig:motif:cbind})}
    The binding of some computational data (pink data-atom, possibly with children).
    \textbf{(\subref{fig:motif:destruct},\subref{fig:motif:destruct-w})}
    The destructuring of a child (blue data-atom) from its parent (pink data-atom).
    In the first case, the identity of the child data-atom (\atom A) is known by the compuzyme, while in the second case it is not.
    \textbf{(\subref{fig:motif:sub-init})}
    The initiation of a subcomputation.
    The pink data-atom (and children) represents some pre-prepared computational data.
    It is bound to a fresh \atom{C} atom, exposing it to the action of other compuzymes.
  }
  \label{fig:motif}
\end{figure}

The RBL atoms for data-atoms and the \atom{C} atom are shown in \Cref{fig:data-mono}.
These data structures can then be operated on by `compuzymes' to perform computation.
Compuzymes are special atoms that behave as a platform for structural manipulation, and are typically represented as a rectangle.
By recruiting data-atoms and interacting with them via their control ports, a sequence of structural manipulations can be effected.
Compuzymes also maintain an internal state, via the color of a pair of self-loops.
These state colors are typically numbered, but any label will suffice.
Internal states are useful for distinguishing between similar configurations, such as may occur during branches, loops, or long sequences.
There are five core compuzyme motifs that are needed to perform arbitrary manipulations.
These are shown in \Cref{fig:motif}, and their detailed implementations may be found in \Cref{app:motif}, along with a description of the bond colorings of the data-ports and control ports.
As RBL is reversible, each of these motifs may also act in reverse; for example, \motif D may also be used to bind a child data-atom to a parent.
By convention, \motif X also expends one unit of fuel to drive computation forward, but in principle fuel coupling can occur at any point(s) during compuzyme operation.

To construct a compuzyme in RBL, we should first prepare an (abbreviated) transition diagram of its operation using the above motifs.
Having done this, we will be able to determine what ports the compuzyme will require.
With this, a suitable RBL atom can be designed.
Then, the RBL implementations of each motif (shown in \Cref{app:motif}) dictate which configurations are possible.
For the most part, each motif will correspond to a distinct compuzyme state, and so the sets of possible configurations will generally be disjoint.
We need merely extend the configurations to include the state of the other compuzyme bonds, which will be static within the motif.
In some cases, particularly with branching control flow, motifs may share configurations.
Thus the complete RBL description of the compuzyme is simply the union of all the configurations of the motifs.

\subsection{Logical Negation}

\begin{figure}
  \begin{subfigure}{\linewidth}
    \centering
    \fig{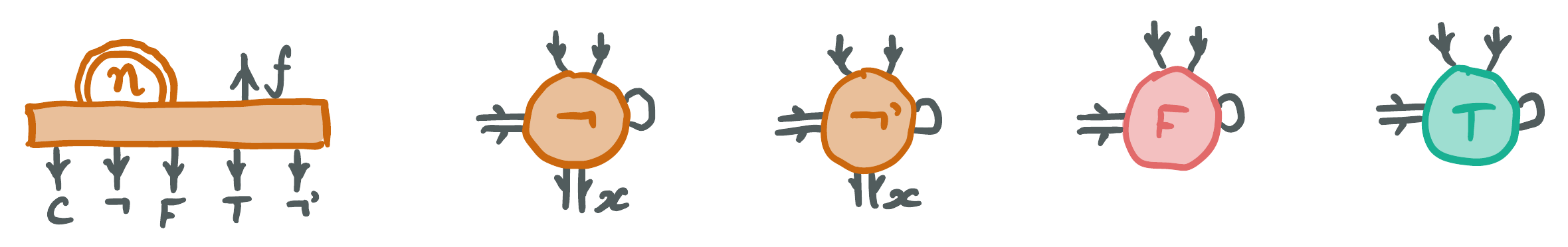}
    \caption{%
      The RBL atoms for logical negation.
      From left to right, the compuzyme \atom{Not}, the computational data-atoms $(\neg)$ and $(\neg')$, and the data-atoms for Boolean values, \atom{F} and \atom{T}.
      \atom{Not} has ports for fuel ($f$), atom \atom{C} ($C$), and the control ports of each of the data-atom types; port labels and types are the same.}
    \label{fig:cnot:mono}
  \end{subfigure}
  \\[0.5em]
  \begin{subfigure}{\linewidth}
    \centering
    \figSm{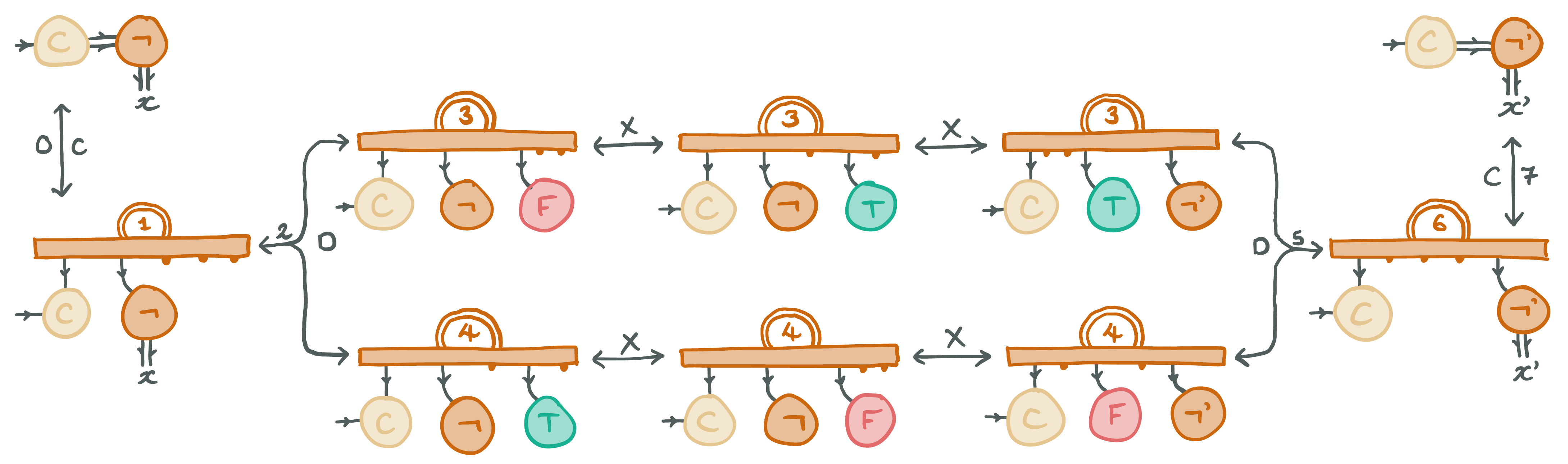}
    \caption{%
      The abbreviated transition diagram for logical negation, showing the motifs employed.
      The compuzyme \atom{Not} starts and ends in its pure unbonded atomic state, and so acts catalytically.
      The $f$ port is not shown, and unbound ports are shown with a dot.
      To aid comprehension, port order is consistent with that shown in (\subref{fig:cnot:mono}).}
    \label{fig:cnot:tx}
  \end{subfigure}
  \caption{%
    The implementation of logical negation using conditional branching.}
  \label{fig:cnot}
\end{figure}

With these motifs we are now in a position to implement our three example programs.
One approach to implement logical negation is to prepare two compuzymes performing each of the two `reactions' in \Cref{fig:rxn-not}, $\atom{Not}_{\atom{T}}$ and $\atom{Not}_{\atom{F}}$ respectively.
The compuzyme implementing $\neg\atom{T}$ would then use \motif C to bind a $(\neg)$ computational data-atom, and then \motif D specialized to \atom{T}.
It can then use \motif X twice to swap \atom{T} for \atom{F} and $(\neg)$ for $(\neg')$, followed by \motif D to bind \atom{F} to $(\neg')$, and finally \motif C to eject the result of the computation.
The compuzyme for $\neg\atom{F}$ would be very similar.
Of course, it is possible that compuzyme $\atom{Not}_{\atom{T}}$ might inadvertently bind \atom{C}\bond$(\neg)$\bond\atom{F}, but the specificity of \motif D would mean the compuzyme stalls.
The only available option would be for it to backtrack; only the second compuzyme can complete the computation in this case.
Moreover, only once the correct input to the compuzyme is bound can it expend fuel.

However, RBL allows us to go beyond this and implement conditional branches in control flow without backtracking.
That is, we can create a single compuzyme \atom{Not} implementing both cases.
We can overlap two (or more) instances of \motif D, each recognizing different data-atoms.
The motifs would use the same states $n$ and $n'$, but differing states $n''$ for each case; these distinct states $n''$ then allow us to distinguish the different branches of control flow.
This is shown in \Cref{fig:cnot}.
Indeed, the dynamics of \atom{Not} are simply the conjunction of those for $\atom{Not}_{\atom{T}}$ and $\atom{Not}_{\atom{F}}$.
Note that, when the two branches of control flow are later merged, this is a reversible operation: the two branches have to be \emph{clearly} distinct.
Here this is the case because we are using \motif D against two cases, \atom{T} and \atom{F}, `combining' the data-atoms into a single variable data-atom $x'$.

\subsection{Addition}

\begin{figure}[ht!]
  \begin{subfigure}{\linewidth}
    \centering
    \fig{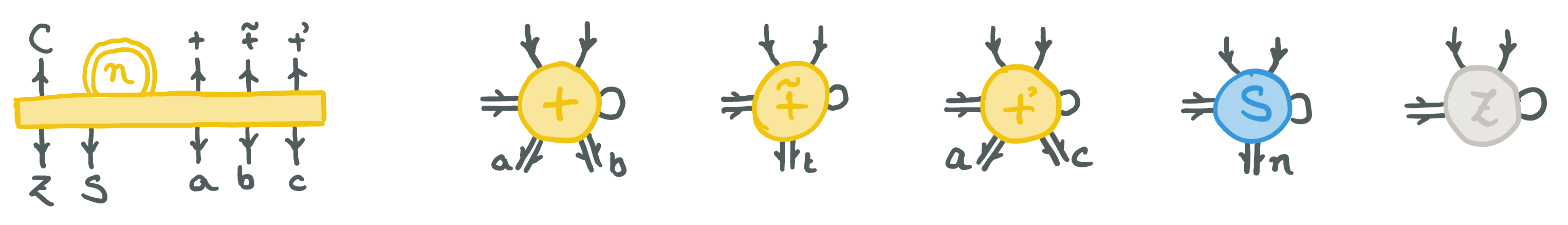}
    \caption{%
      The RBL atoms for addition.
      From left to right, the compuzyme \atom{Add}, the computational data-atoms $(+)$, $(\tilde +)$, and $(+')$, and the data-atoms for Peano numbers, \atom{S} and \atom{Z}.
      \atom{Add} has ports for fuel ($f$, not shown), atom \atom{C}, and control ports for each type of data-atom, but also wildcard control ports $a$, $b$, and $c$ for holding variable values during computation; port labels and types are the same, except for the variables which are of type $\ast$.}
    \label{fig:add:mono}
  \end{subfigure}
  \\[0.5em]
  \begin{subfigure}{\linewidth}
    \centering
    \figSm{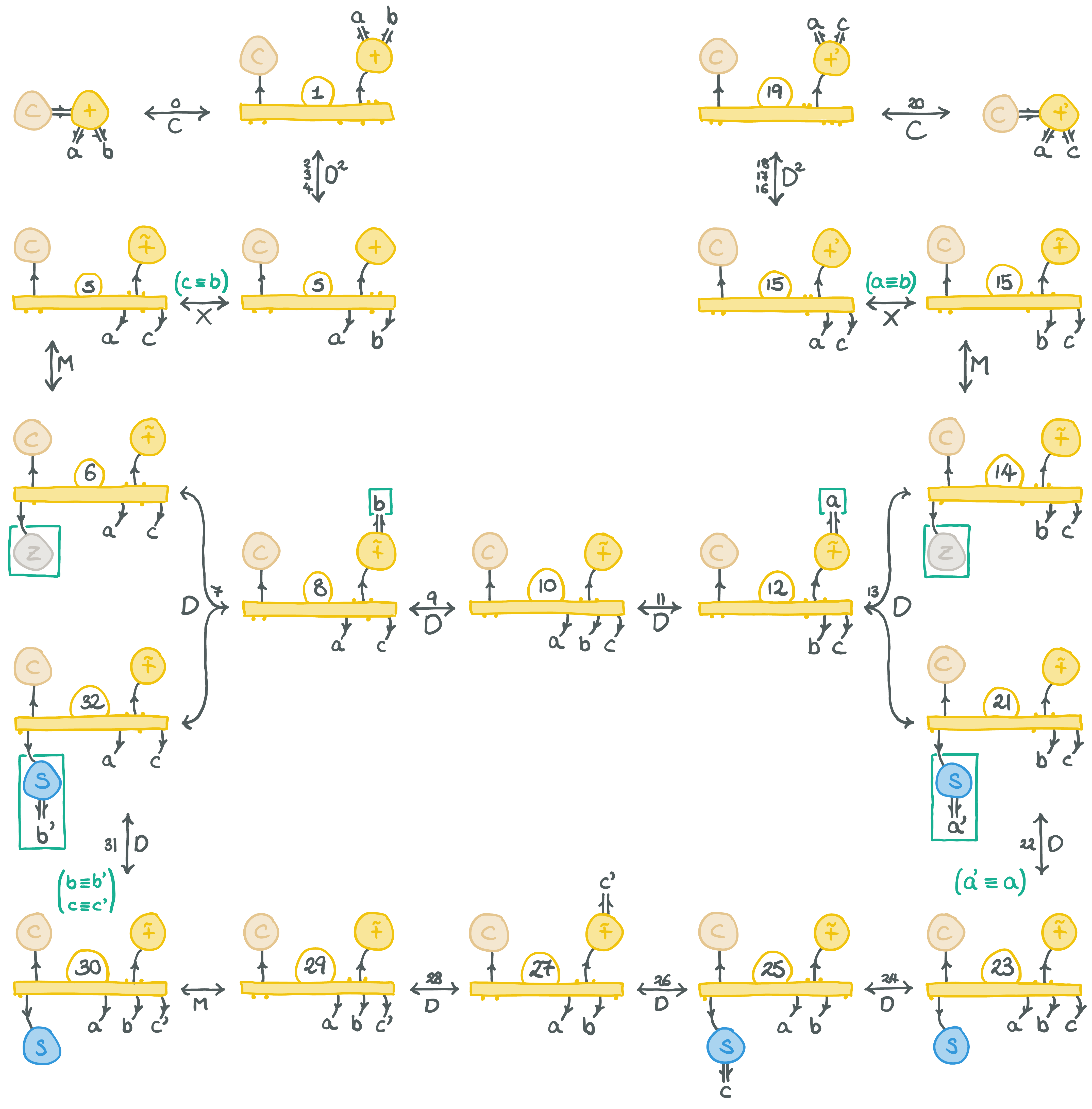}
    \caption{%
      The abbreviated transition diagram for addition, showing the motifs employed.
      Sometimes, two \motif Ds are elided into a single transition (shown as $\mathsf{D}^2$).
      Green annotations show where variables are renamed between transitions or where two alternative data-atoms are bound to a variable position.
      To aid comprehension, port order is consistent with that shown in (\subref{fig:add:mono}) except that atoms above the compuzyme are flipped vertically; additionally, unbound ports are shown with a dot, and the compuzyme state is shown with a single circle.
      Addition is implemented as a loop where the value of $a$ is added to both $b$ (renamed to $c$, and yielding $c=a+b$) and $0$ (yielding a copy of $a$).
      The top left set of transitions enter the loop.
      The center set of transitions implement the control flow of the loop:
        on the left, they merge branches corresponding to loop entry and loop continuation, using $b$ to discriminate these branches;
        on the right, they branch into the exit path or the main loop body, using $a$ to discriminate these branches;
        the middle transitions tidy up/set up these branch operations.
      The bottom set of transitions implement the loop body, decrementing $a$ and incrementing both $b$ (now representing the copy of $a$) and $c$.}
    \label{fig:add:tx}
  \end{subfigure}
  \caption{The implementation of addition using looping.}
  \label{fig:add}
\end{figure}

Similarly, we could implement addition using three compuzymes as indicated in \Cref{fig:rxn-add}.
However, we again leverage the power of RBL to show how looping can be implemented directly with a single compuzyme \atom{Add}.
This implementation is given in \Cref{fig:add}.
Note that in this condensed implementation, the atom $(\tilde +)$ now just serves to hold temporary variables and so only has one child compared to three in our original schematic.
Looping is not much different from the control flow for logical negation;
  whereas for conditional branching, we had a branch followed by a merge, here we have a merge followed by a branch.
At the beginning of a loop, we merge control flow from two branches: entry and loop-continuation.
This would usually be trivial, but in a reversible context we need to be careful to conditionally distinguish these two events (so that, in reverse, we know when to `un-start' the loop).
At the end of a loop, we branch into exit and loop-continuation.

In our implementation of addition, we start with $(a,b)$ and end with $(a,c)$ where $c=a+b$.
We begin by renaming $b$ to $c$, to which we will add the value of $a$.
Then we set the now free variable $b$ to 0, to which we will also add the value of $a$.
In our loop, the entry condition is therefore that the accumulator ($b$) is 0; if it is greater, then we have performed at least one iteration.
The exit condition is that $a=0$, as we decrement it on each iteration.
Therefore, we merely need to use \motif D to check whether these variables are \atom{Z} or \atom{S}\bond$n$.

Due to reversibility, this program can also perform subtractions.
To achieve this, one could either reverse the bias of the free energy currency, duplicate the implementation with reversed fuel coupling, or introduce a `switch' for the direction of computation.

\subsection{Squaring}

\begin{figure}[ht!]
  \begin{subfigure}{\linewidth}
    \centering
    \fig{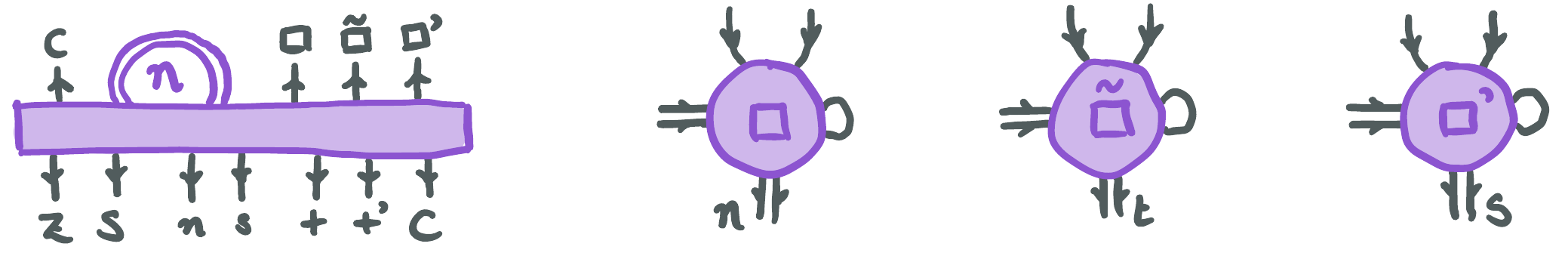}
    \caption{%
      The RBL atoms for squaring.
      From left to right, the compuzyme \atom{Sq} and the computational data-atoms $(\square)$, $(\tilde\square)$, and $(\square')$.
      \atom{Sq} has ports for fuel ($f$, not shown), atom \atom{C}, and control ports for the computational data-atoms shown, Peano data-atoms, two variables ($n$ and $s$), as well as a set of ports for the addition subroutine: the computational data-atoms $(+)$ and $(+')$ and the atom \atom{C}; port labels and types are the same, except for the variables which are of type $\ast$.}
    \label{fig:sq:mono}
  \end{subfigure}
  \\[0.5em]
  \begin{subfigure}{\linewidth}
    \centering
    \fig{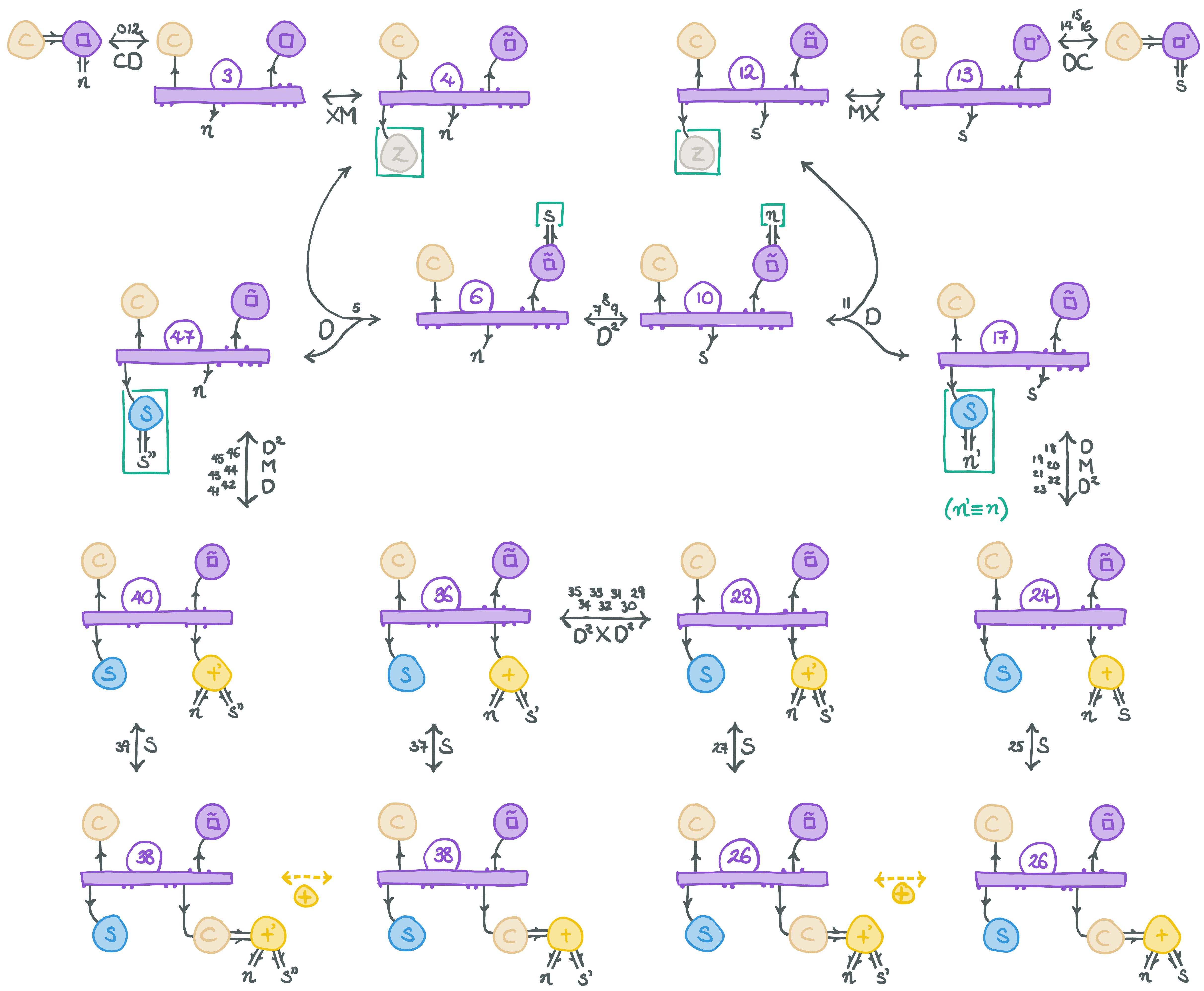}
    \caption{%
      The abbreviated transition diagram for squaring, showing the motifs employed.
      Some motifs are elided into a single transition for brevity.
      Green annotations show where variables are renamed between transitions or where two alternative data-atoms are bound to a variable position.
      To aid comprehension, port order is consistent with that shown in (\subref{fig:add:mono}) except that atoms above the compuzyme are flipped vertically; additionally, unbound ports are shown with a dot, and the compuzyme state is shown with a single circle.
      The implementation of squaring consists of a loop which on each iteration maps $s\mapsto s+(2n-1)$ and $n\mapsto n-1$, with $s$ initially 0 and $n$ ending at 0.
      The loop structure is essentially the same as for addition (\Cref{fig:add}).
      To perform the map $s\mapsto s+(2n-1)$, we first decrement $n$, then add the new value of $n$ to $s$ twice before incrementing it.
      The additions are performed using \atom{Add} by preparing computational data representing the two additions.}
    \label{fig:sq:tx}
  \end{subfigure}
  \caption{The implementation of squaring using subroutines.}
  \label{fig:sq}
\end{figure}

Finally, we implement squaring.
Again, squaring consists of a loop;
  however, this loop involves calling addition as a subroutine.
This is quite simply done.
We use the data manipulation motifs to bind the numbers we wish to add to a $(+)$ atom;
  then \motif S binds this computational data to a \atom{C} atom, thus `presenting' it to \atom{Add} compuzymes.
The implementation is shown in \Cref{fig:sq}.

As with addition, the reverse of this program is also useful and performs square roots (over a suitably restricted domain---attempting to take the square root of 10 would lead to something resembling a runtime error).

\subsection{Reversible Turing Machines}

The implementation of reversible Turing Machines (defined by Bennett~\cite{bennett1973logical}) is left as an exercise for the reader.
It only requires conditional branching, but the reversible implementation of a bi-infinite tape requires some care.
For a hint, a reference implementation in a related language is given in the author's thesis~\cite{earley2021performance} (p.~147).

\section{Discussion}
\label{sec:concl}

As required, the RBL model admits a variety of forms of component reuse:
  monomers used in the construction of molecules can be fully recycled, being drawn from and returned to a pool of free components;
  dynamic components can be designed to act catalytically, and so may be reused multiple times;
  and the same monomers (e.g.\ those representing Peano numbers) can be reused in distinct contexts without interference.
These properties are strengthened by the modularity conferred by molecular structure:
  any number of instances of the same program or system can run in parallel using the same components and without crosstalk;
  and common sub-systems can be factored out and used by multiple distinct systems, such as a single common fuel species supplying all the free energy in the system.
Moreover, RBL is particularly amenable to implementing hierarchies of abstraction.
For instance, we saw the following hierarchy of computational abstractions:
  the definition of RBL `data-atoms';
  a set of basic structural-manipulation motifs;
  control flow primitives such as branching and looping;
  and subroutines such as addition, which could then be used as a black-box by the squaring routine.
Further abstractions which are possible but not shown include recursion and concurrency.

While RBL is a powerful model, it also has some limitations in its current formulation.
Simulations of the squaring routine encounter 1700 distinct configurations when computing $3^2$ (alternatively, $\sqrt{9}$).
This is a reflection of the hierarchy of abstractions used; the structural manipulation motifs incur a cost of $\sim\mbox{10--15}$ transitions each, and multiple motifs must be employed to perform a given structural manipulation leading to an overhead on the order of 50--100 transitions for each useful computational step.
Simpler `machines' such as the biased walker lie closer to `pure' RBL and incur 8--14 transitions per step (depending on how we count the parallel transitions), though this is still perhaps larger than desired.
Arguably the most significant reason for these overheads originates in the lack of spatial/geometric awareness of RBL.
In biochemical systems, enzymes are able to take advantage of shape complementarity and differently-shaped conformations to perform their manipulations.
This description omits the series of microstate transitions involved in an enzymatic reaction, but even so enzymes are (usually) particularly fast and efficient.
Meanwhile in RBL we must very carefully coordinate structural manipulations through the use of, e.g., control ports: our lack of spatial awareness means that, without such coordination, we could not selectively displace a particular child data-atom from its parent.
This suggests an obvious RBL variant to pursue in the future: one in which geometry plays a first-class role.
Such a model should be able to more directly perform desired structural manipulations.
Moreover, incorporating geometry would ensure that our designs are physically possible: the non-geometric RBL model can easily encode structures that are too crowded for Euclidean space, such as a complete binary tree of depth 20.

Another challenge for the experimental realization of RBL is that the number of configurations of an RBL atoms scales exponentially with the number of ports, and each of these may have a distinct energy.
Restricting the set of allowed energies $E$ to $\{0,\infty\}$ would considerably simplify this, as we may be able to get away with just programming the allowed configurations.
Nevertheless, compuzymes such as \atom{Add} and \atom{Sq} have on the order of 200 allowed configurations.
Furthermore, control ports and compuzyme states involve dozens of possible bond colors.
Future work should investigate whether the number of required bond colors can be reduced, or whether bond colors are even required (i.e.\ whether ports can be monochromatic).
Another question is whether atom configurations can be `factored' into simpler subsystems.
Additionally, it would be useful to know what the minimum required complexity is to implement useful computational abstractions.
For example, while \atom{Add} may require $\sim 200$ configurations, the division into three compuzymes $\atom{Add}_{\mathrm{init.}}$, $\atom{Add}_{\mathrm{loop}}$, and $\atom{Add}_{\mathrm{fin.}}$ may be substantially simpler.
Combining with a geometric variant of RBL should also reduce the number of distinct configurations required.

Lastly, a more accurate treatment of the kinetics of RBL is warranted.
Not only will configuration energy contribute to transition rates, but also any associated entropy changes.
These entropy changes are particularly significant when particle number changes, such as when a compuzyme binds some computational data.
There are also entropy changes associated with interactions with the monomer pools that cannot be eliminated through any `clever' means (see the author's thesis~\cite{earley2021performance} (pp.~101--120)).

There is some similarity between RBL and Thermodynamic Binding Networks (TBNs)~\cite{doty2017thermodynamic}.
TBNs consist of monomers with a geometry-free collection of domains, similar to RBL atoms with a geometry-free collection of ports.
Domains have complementary codomains, and these can form bonds.
However, the TBN model explicitly does not consider the kinetics of system evolution.
Instead it focuses on analyzing the possible stable configurations admitted by a given TBN to determine whether leak reactions are possible, and presumes that the desired state corresponds to thermodynamic equilibrium.
In this way, the goals of TBNs diverge from the goal of RBL.
In RBL, it is the programming of kinetic pathways that is important, and equilibration is to be avoided.
Nevertheless, a realistic implementation of RBL would likely be subject to error conditions such as leak reactions, and so it would be interesting to use a TBN-like framework to evaluate or improve the robustness of RBL systems.

In conclusion, RBL is an interesting new model with powerful properties, but it is also currently too unwieldy for experimental realization.
Future work will further develop RBL, including variant models, to determine the possibility of achieving these powerful properties in a real chemical system.

\bibliography{references}
\appendix
\input{appendix}

\end{document}

%% file: appendix.tex
\section{Proofs}
\label{app:prf}

\biasRandomWalk*
\begin{proof}
  Recall the definition of the biased walker (\Cref{sec:walk}) and its transition diagram (\Cref{fig:walk-ex}).
  We will assume typical CRN dynamics, namely that the evolution of the system may be described by a Continuous-Time Markov Chain (CTMC) with transition rates given by the law of mass action.
  For example, the reaction $W_\alpha + G_+ \rightleftharpoons W_\alpha{:}G_+$ has forward transition rate $k_2[W_\alpha][G_+]$ and reverse transition rate $k_1[W_\alpha{:}G_+]$ where $[X]$ is the concentration of species $X$, $k_1$ is the rate constant for unimolecular reactions, and $k_2$ that for bimolecular reactions (all the species have equal energy, and we assume no activation energy).
  Note the periodic symmetry of the walker scheme;
  without loss of generality, we can relabel the transition diagram to go from $W_n$ to $W_{n+1}$ where $n$ corresponds to the $n^{\text{th}}$ position along the track.
  We assume the track extends infinitely in both directions, so that $n\in\mathbb Z$.
  We first prove that the net reaction $W_n \rightleftharpoons W_{n+1}$ has rate constants $\propto [G_+]$ and $\propto [G_-]$ for the forward and reverse reactions.

  The dynamics of the walker are somewhat complicated by the number of intermediate states---13 between each net step.
  However, after some convergence time, any initial distribution will converge to a steady state distribution (up to periodicity of the Markov Chain (MC)); this is analogous to the steady state approximation in chemistry.
  To see this, note that the behavior within each step-window (between $W_n$ and $W_{n+1}$) is identical.
  Therefore, we can overlap each of these `sub-chains', reducing the infinite MC to the finite MC from $W_n$ to $W_{n+1}$ by symmetry.
  As this MC is aperiodic, irreducible, and reversible, it admits a steady state given by detailed balance and this steady state is reached exponentially fast.
  After this `burn-in' time, the initial distribution is effectively forgotten.

  For any pair of adjacent intermediate states $i$ and $j$ (states of the form $W{:}G$), the transition rate for $i\to j$ is $k_1[i]$.
  Consequently, by detailed balance their steady state concentrations are all equal (within a given step-window), and we can write this concentration as $[W{:}G]$.
  The remaining reactions are $W_n + G_+ \rightleftharpoons W_n{:}G_+$ and $W_{n+1}{:}G_- \rightleftharpoons W_{n+1} + G_-$.
  Using the transition rates given earlier, $k_2[W_n][G_+] = k_1[W{:}G]$ and $k_1[W{:}G]=k_2[W_{n+1}][G_-]$, and hence $[W_{n+1}]/[W_n] = [G_+]/[G_-]$.
  Furthermore, the net forward reaction rate constant for $W_n \rightleftharpoons W_{n+1}$ is $k_2[G_+]$, and the reverse reaction rate constant is $k_2[G_-]$.

  Having proved this, we can reduce the original CTMC to the effective CTMC
  \begin{center}\small
  \begin{tikzpicture}[scale=0.7]
    \node[circle,draw,minimum size=1.2cm] (n1) at (0,0) {$W_{n-1}$};
    \node[circle,draw,minimum size=1.2cm] (n2) at (4,0) {$W_n$};
    \node[circle,draw,minimum size=1.2cm] (n3) at (8,0) {$W_{n+1}$};

    \draw [->] (-2,0.7) to [out=0,in=150] node [midway,above] {} (n1);
    \draw [->] (n1) to [out=30,in=150] node [midway,above] {$k_2[G_+]$} (n2);
    \draw [->] (n2) to [out=30,in=150] node [midway,above] {$k_2[G_+]$} (n3);
    \draw [->] (n3) to [out=30,in=180] node [midway,above] {} (10,0.7);

    \draw [->] (10,-0.7) to [out=180,in=330] node [midway,below] {} (n3);
    \draw [->] (n3) to [out=210,in=330] node [midway,below] {$k_2[G_-]$} (n2);
    \draw [->] (n2) to [out=210,in=330] node [midway,below] {$k_2[G_-]$} (n1);
    \draw [->] (n1) to [out=210,in=0] node [midway,below] {} (-2,-0.7);
  \end{tikzpicture}
  \end{center}
  which is prototypical of a random walk with bias $b=([G_+]-[G_-])/([G_+]+[G_-])$.
  The expected value of $n$ at time $t$, given $n(0)=0$, will be $k_2bt$.
  The variance is non-trivial for a continuous-time process, but will be approximately $\sqrt{k_2t[G_+][G_-]}/([G_+]+[G_-])$.
\end{proof}

\section{Data \& Compuzyme Motifs}
\label{app:motif}

Missing from the description of data-atoms and data-ports in \Cref{sec:motif} are the bond colors and atom configurations.

Data-ports consist of two ports.
The first port (counting clockwise) has two bond colors, solid and dashed/$\tss$.
The second port is monochromatic.
This design allows for the controlled and coordinated breaking of bonds.

Control ports have more colors to allow for the intricate signaling required by the motifs.
These colors are:
  $m$, for `monomer';
  solid/neutral, a bond that has no signaling intent;
  dashed/$\tss$, for transitional states;
  $\atom{C}\bullet$, $\atom{C}\tss$, and $\atom{C}\circ$,
    for displacing a data-atom from its \atom C parent;
  and $x\bullet$, $x\tss$, and $x\circ$ for each child data-port $x$, for displacing child data-atoms.
$\bullet$, $\tss$, and $\circ$ indicate that the relevant entity is bound, is transitioning between bound and unbound, or is unbound, respectively.

Wildcard control ports are much simpler, having just two bond colors: solid and dashed/$\tss$.

With the bond colors enumerated, we now define the allowed RBL configuration for data-atoms and compuzymes.
These are illustrated piecemeal in the remainder of this Appendix.

\begingroup
\subsection*{State changes}

Not a motif in itself, compuzyme state changes are common to almost all our motifs.
Compuzymes may perform a long series of manipulations, including branched control flow.
As a result, indistinguishable configurations may be encountered, unexpectedly linking distinct parts of the computation.
To prevent this, we imbue compuzymes with an internal `state' to track progress through configuration-space and correctly handle control flow.
The implementation is below:
\begin{align*}
  \begin{gathered}
    \figSm{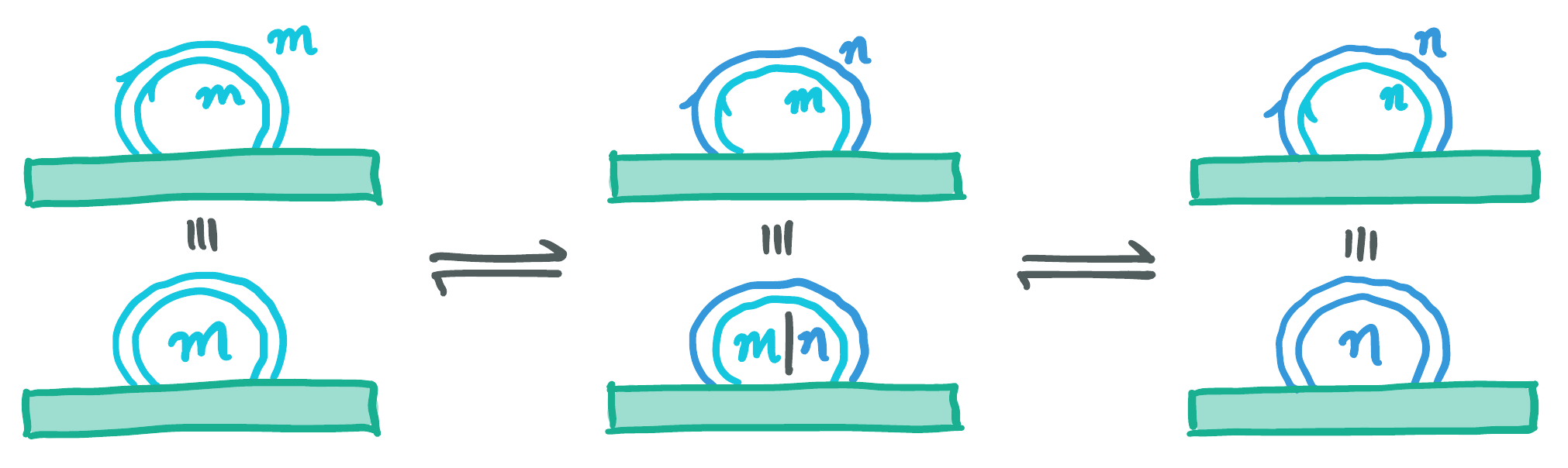}
  \end{gathered}
\end{align*}
The top row is an RBL schematic, whereas the bottom row presents an abbreviated notation.
The state is represented by a pair of self-loops.
Each of these have the same set of bond colors, with each color representing a distinct state (e.g.\ $m$, $n$).
Note that $m$ and $n$ are variables rather than explicit colors, i.e.\ $m$ is not the same as the monomer color.
Typically these states will be numbered, but any set of useful labels may be employed.
There is also a special `$\varnothing$' state in which the bonds are broken; compuzymes typically start and end in this state.
Changing state from $m$ to $n$ is simply a matter of changing the bond colors of the loops in turn.

The requirement of a pair of loops may seem superfluous;
however, consider the case of two otherwise identical configurations that are meant to be distinguished by this state.
If a single loop was used, then these configurations would be adjacent and so the state would not serve as a barrier.
The double loop prevents this scenario, as the intermediate $m|n$ state will be constructed so as to be inaccessible.

\endgroup

\begingroup
\renewcommand{\thesubsection}{M}
\subsection{Monomer binding}

Data molecules are formed from data-atoms, and so the manipulation of individual data-atom `monomers' is foundationally important.
When building up data, fresh monomers will need to be drawn; conversely, when breaking down data the extraneous monomers will need to be discarded.
We suppose that the environment provides an unlimited pool of fresh monomers for these purposes.
Both drawing and discarding monomers can be implemented by a single motif, as they are inverses of each other and our system is reversible:
\begin{align*}
  \begin{gathered}
    \figSm{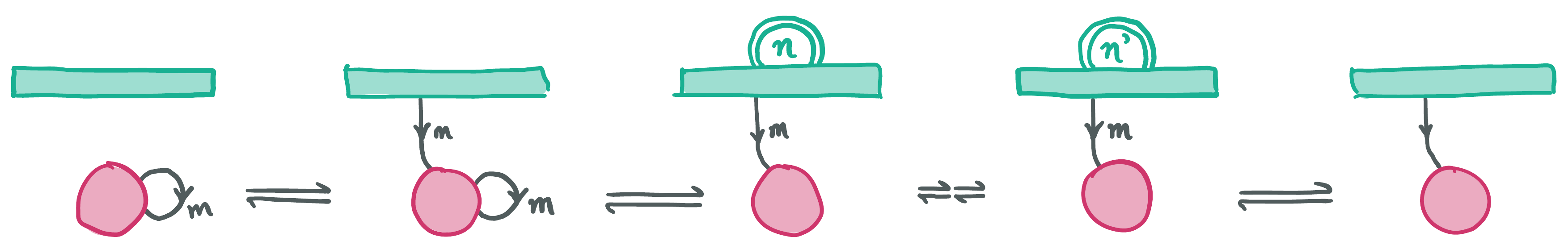}
  \end{gathered}
\end{align*}
Here the monomer is the pink data-atom, and the compuzyme drawing (resp.\ discarding) the monomer is the green rectangle.
In general we consider compuzymes as a platform for manipulations, hence the rectangular form.
Both the free compuzyme and free monomer have an adjacent configuration in which the control ports form an `$m$-colored' bond.
Consequently, the compuzyme is able to spontaneously form a bond with a free monomer of the correct type, drawing it from solution.
Going forward, we will not explicitly mention adjacent configurations when they can be inferred from the diagrams.
Recall that control ports are unique to a given data-atom type.
As such, a compuzyme can be sure it is drawing the desired monomer from solution.
We conclude the motif by switching the control port to the `neutral' bond color, ready for subsequent transitions.

Notice that the monomer in solution has a self-loop, also $m$-colored.
Suppose that it didn't and instead had 0 bonds.
In this case, the neutral bond in the final configuration could readily break.
To address this, we assign an infinite energy to data-atoms with 0 bonds, and use the self-loop to mark monomers.
Observe further that the compuzyme has no direct way to break the self-loop on the monomer.
This is the reason for the existence of control ports: through different bond colors, we can signal different intents.
Upon receiving such a signal, the normally-inert data-atom is able to transition to other configurations.
Specifically, there are usually $E=\infty$ energy barriers preventing a data-atom from transitioning to other configurations; but for specific control port colors, these energy barriers are lowered and a small region of configuration space is made available to explore.
Note also the state change, the purpose of which we will elaborate further in \motif C.

\endgroup

\begingroup
\renewcommand{\thesubsection}{X}
\subsection{Monomer exchange}

Related to the monomer binding motif is the monomer exchange motif.
Arguably this is not a single motif in its own right, as it can be realized by combining two monomer-binding motifs back-to-back.
However, monomer exchange is sufficiently common and useful that we promote it to its own motif.
The implementation follows:
\begin{align*}
  \begin{gathered}
    \hspace{-1em}\figSm{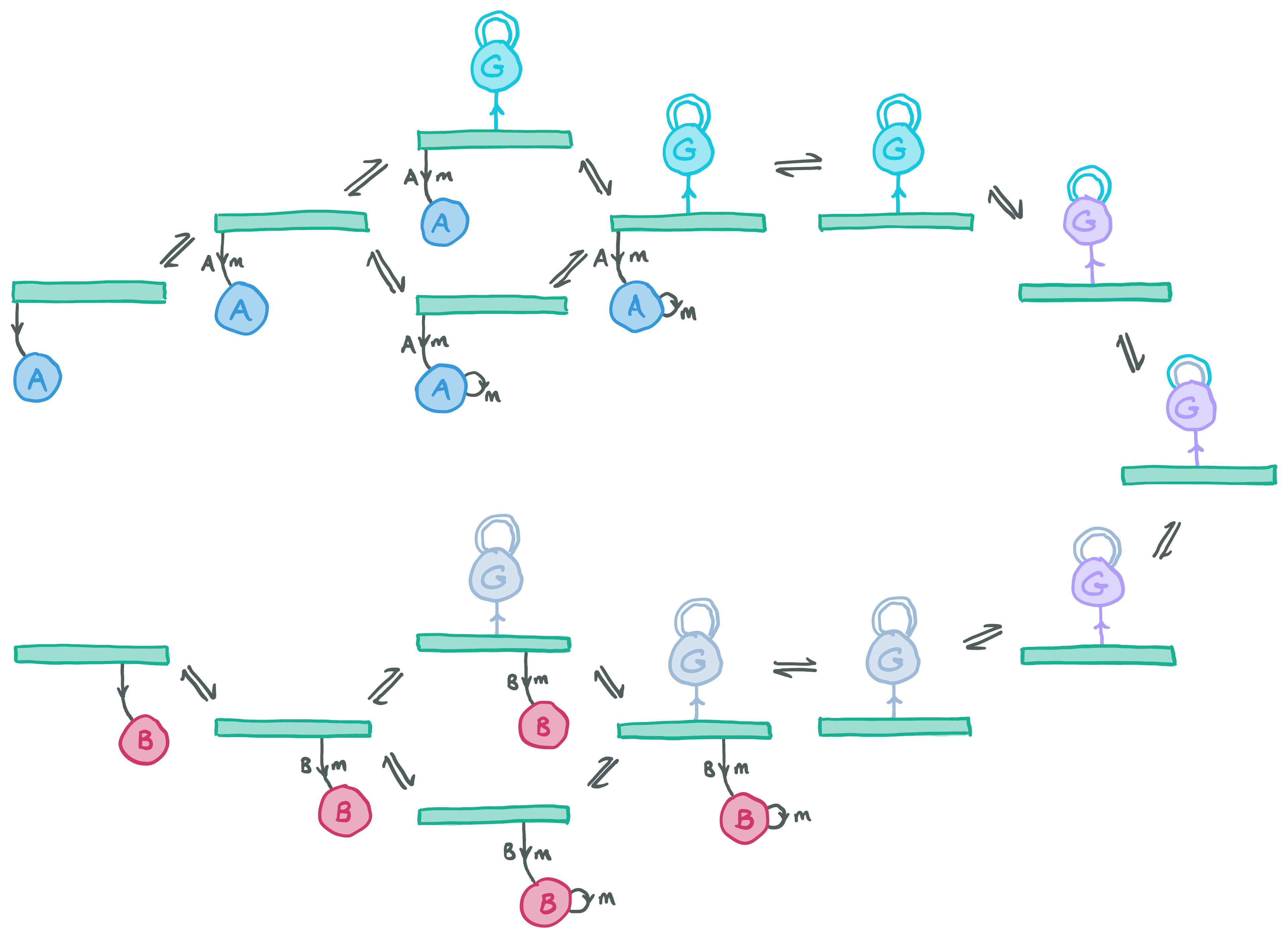}
  \end{gathered}
\end{align*}
Notice that its implementation differs from two occurrences of \motif M:
  the state changes have been replaced by coupling to fuel.
Though in principle any motif or state change can be coupled to the burning of fuel, we choose this motif as the canonical such point for convenience.

\endgroup

\begingroup
\renewcommand{\thesubsection}{C}
\subsection{Computational data binding}

To bind computational data in order to manipulate it we need to (1) recognize it as such via its control port, and (2) displace it from its \atom{C} parent:
\begin{align*}
  \begin{gathered}
    \hspace{-1em}\figSm{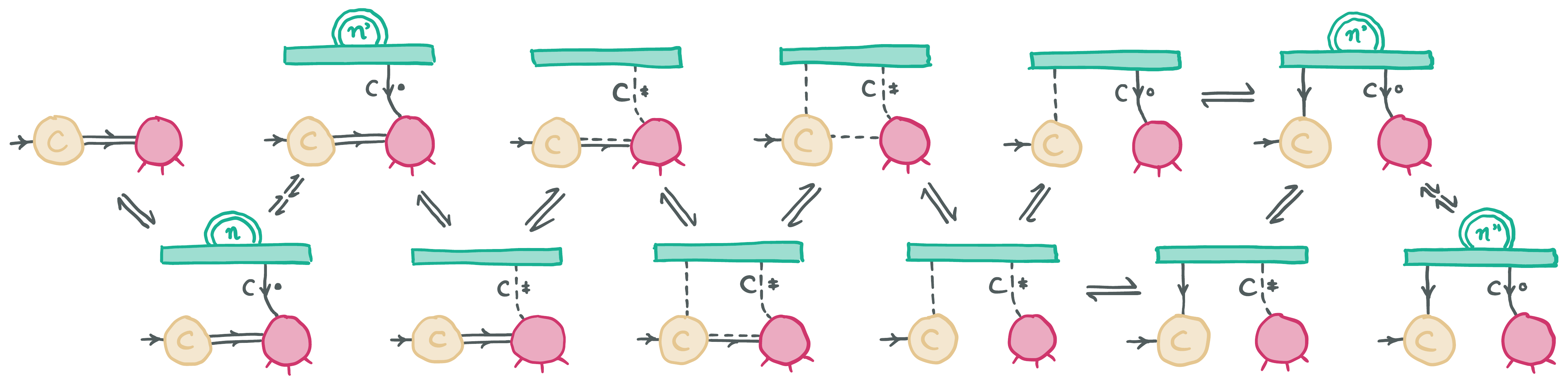}
  \end{gathered}
\end{align*}
The \atom{C} atom is in beige, and the computational atom in pink.
Recall that the control port has three colors indicating progress of binding and displacement:
  $\atom{C}\bullet$ corresponds to computational data with a \atom{C} parent while $\atom{C}\circ$ corresponds to computational data displaced from its parent.
  $\atom{C}\tss$ is a `transitional' state between these.
We bind with $\atom{C}\bullet$, and then immediately change state from $n$ to $n'$.
These state changes are frequent in motifs and act as \emph{bidirectional assertions}.
Without this, we could---for example---break the control port-bond during a transitional state.
After this assertion, we switch to the transitional bond color $\atom{C}\tss$.
In this configuration, the bond between computational atom and \atom{C} parent can be weakened in a number of steps until it's broken completely.
At the same time, we must form a bond with the soon-to-be-displaced \atom{C} atom so as not to lose the association between them.
This association is important, because this computation may be a subcomputation, with the particular \atom{C} atom bound by another compuzyme.
If we were to lose the association, then we could not return the result of this computation to the caller.
Indeed, the series of adjacent configurations is programmed carefully so that the bond between \atom{C} and the computational atom cannot be completely broken until the compuzyme has bound both.
Finally we have another bidirectional assertion from state $n'$ to $n''$, after which further transitions and motifs may occur.
Notice that two of the last bond changes can happen in either order, hence the parallelism in the transition diagram.

\endgroup

\begingroup
\renewcommand{\thesubsection}{D}
\subsection{Destructuring}

Data molecules are tree structures formed from data-atoms.
To effectively manipulate these, we need to be able to break down (and build up) these structures into their constituent parts.
We achieve this with the destructuring motif, which can be used to displace (resp.\ replace) a child data-atom from its parent.
This motif is nearly identical to that of binding computational data, primarily differing in the control port bond colors used:
\begin{align*}
  \begin{gathered}
    \hspace{-1em}\figSm{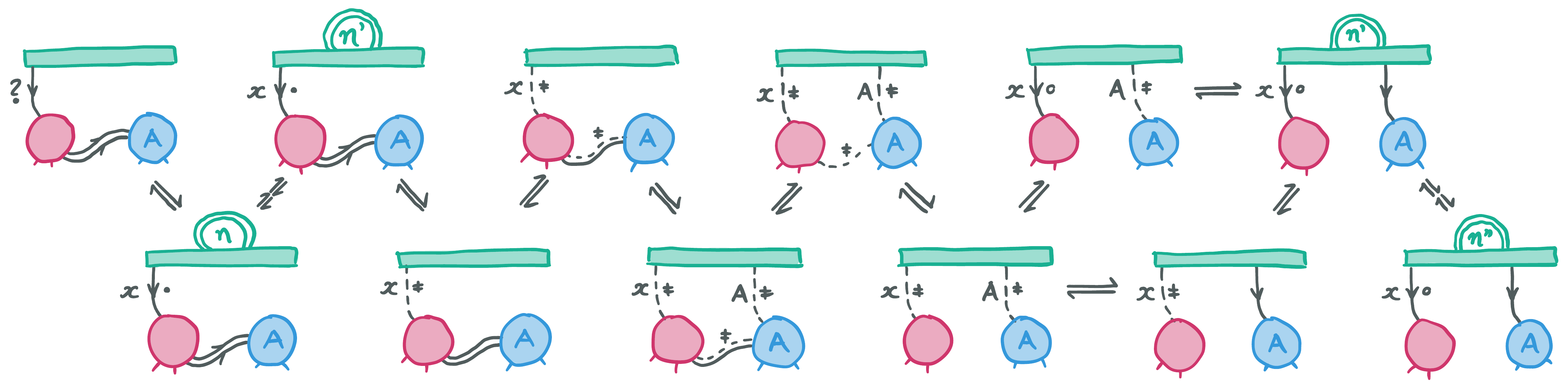}
  \end{gathered}
\end{align*}
The $?$ bond `color' indicates that this motif doesn't care what the initial state of the control port is.
Recall that a data-atom may have multiple children, each bound by a differently-labelled data-port.
Here we are displacing a child data-atom (blue) from the $x$ data-port of its parent (pink).
Specifically, the child data-atom is of type \atom{A}, and indeed this particular instance of \motif D can \emph{only} displace an atom of type \atom{A}.
Note that the child data-atom may have its own children.

If we wish to introduce a branch (resp.\ merge) in control flow, then multiple of these motifs can be overlaid---one for each type of data-atom to be recognized.
The first 5 configurations shown are common, as they are independent of the child data-atom type.
The remaining configurations diverge.
The nature of configuration adjacency ensures that this works as expected.
We may also want to displace a data-atom regardless of its type.
For this, we substitute the interactions with the control port for the wildcard control port.

\endgroup

\begingroup
\renewcommand{\thesubsection}{S}
\subsection{Subcomputation initiation}

Because of the indirection provided by the \atom{C} atom, compuzymes act independently of the context of a computation.
This means that any routines we implement can also be used as subroutines.
To call a subroutine, we prepare a `subcomputation':
  computational data that is bound to a host compuzyme and presented on a \atom{C} atom.
Other relevant compuzymes can then act on this subcomputation.
Upon completion, the host compuzyme can accept the result and use it as required.

Prior to using this motif, the appropriate data representing the subcomputation should be prepared by means of the other motifs.
Then, the subcomputation initiation motif will bind it to a fresh copy of the \atom{C} atom:
\begin{align*}
  \begin{gathered}
    \hspace{-1em}\figSm{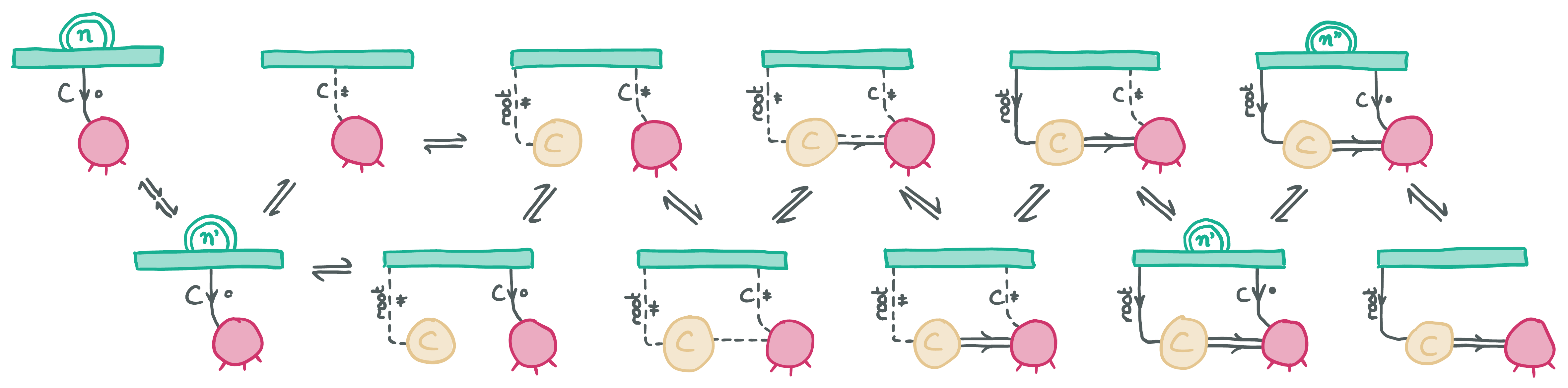}
  \end{gathered}
\end{align*}
This motif is specialized to the type of the data-atom.
To complete the subroutine, we use another instance of this motif in reverse and specialized to the result data-atom.
For example, if the initial data is tagged with $(+)$, then the final data would be tagged with $(+')$.

Now the purpose of the \atom{C} atom becomes clear.
  We need to release the computational atom's control port so other compuzymes can interact with it, but then we would lose track of it.
  Therefore, we need to maintain a separate bond to this molecule.
  However, we cannot simply bind the data-atom on another port: it is important that this data-atom can be swapped out, for example $(+)$ for $(+')$, so as to indicate the completion of computation.
  As such, we use the intermediate atom \atom{C} to achieve the necessary indirection to satisfy all of these conditions.
In particular, we bind the \atom{C} atom via its `root' port; its $c$ port can then be bound by other compuzymes as needed.

%% file: paper.bbl
\begin{thebibliography}{10}

\bibitem{adleman1994molecular}
Leonard~M Adleman.
\newblock Molecular computation of solutions to combinatorial problems.
\newblock {\em science}, 266(5187):1021--1024, 1994.

\bibitem{bennett1973logical}
Charles~H Bennett.
\newblock Logical reversibility of computation.
\newblock {\em IBM journal of Research and Development}, 17(6):525--532, 1973.

\bibitem{bennett1982thermodynamics}
Charles~H Bennett.
\newblock The thermodynamics of computation—a review.
\newblock {\em International Journal of Theoretical Physics}, 21:905--940,
  1982.

\bibitem{brailovskaya2019reversible}
Tatiana Brailovskaya, Gokul Gowri, Sean Yu, and Erik Winfree.
\newblock Reversible computation using swap reactions on a surface.
\newblock In {\em DNA Computing and Molecular Programming: 25th International
  Conference, DNA 25, Seattle, WA, USA, August 5--9, 2019, Proceedings 25},
  pages 174--196. Springer, 2019.

\bibitem{del2022dissipative}
Erica Del~Grosso, Elisa Franco, Leonard~J Prins, and Francesco Ricci.
\newblock Dissipative dna nanotechnology.
\newblock {\em Nature Chemistry}, 14(6):600--613, 2022.

\bibitem{doty2017thermodynamic}
David Doty, Trent~A Rogers, David Soloveichik, Chris Thachuk, and Damien Woods.
\newblock Thermodynamic binding networks.
\newblock In {\em DNA Computing and Molecular Programming: 23rd International
  Conference, DNA 23, Austin, TX, USA, September 24--28, 2017, Proceedings 23},
  pages 249--266. Springer, 2017.

\bibitem{earley2021performance}
Hannah~A Earley.
\newblock {\em On the performance and programming of reversible molecular
  computers}.
\newblock PhD thesis, University of Cambridge, 2021.

\bibitem{eshra2017dna}
Abeer Eshra, Shalin Shah, and John Reif.
\newblock Dna hairpin gate: A renewable dna seesaw motif using hairpins.
\newblock {\em arXiv preprint arXiv:1704.06371}, 2017.

\bibitem{evans2017physical}
Constantine~G Evans and Erik Winfree.
\newblock Physical principles for dna tile self-assembly.
\newblock {\em Chemical Society Reviews}, 46(12):3808--3829, 2017.

\bibitem{kishi2018programmable}
Jocelyn~Y Kishi, Thomas~E Schaus, Nikhil Gopalkrishnan, Feng Xuan, and Peng
  Yin.
\newblock Programmable autonomous synthesis of single-stranded {DNA}.
\newblock {\em Nature chemistry}, 10(2):155--164, 2018.

\bibitem{landauer1961irreversibility}
Rolf Landauer.
\newblock Irreversibility and heat generation in the computing process.
\newblock {\em IBM journal of research and development}, 5(3):183--191, 1961.

\bibitem{lutz1986janus}
Christopher Lutz and Howard Derby.
\newblock Janus: a time-reversible language.
\newblock {\em Letter to R. Landauer}, 2, 1986.

\bibitem{montagne2011programming}
Kevin Montagne, Raphael Plasson, Yasuyuki Sakai, Teruo Fujii, and Yannick
  Rondelez.
\newblock Programming an in vitro {DNA} oscillator using a molecular networking
  strategy.
\newblock {\em Molecular systems biology}, 7(1):466, 2011.

\bibitem{padilla2014asynchronous}
Jennifer~E Padilla, Matthew~J Patitz, Robert~T Schweller, Nadrian~C Seeman,
  Scott~M Summers, and Xingsi Zhong.
\newblock Asynchronous signal passing for tile self-assembly: Fuel efficient
  computation and efficient assembly of shapes.
\newblock {\em International Journal of Foundations of Computer Science},
  25(04):459--488, 2014.

\bibitem{qian2011efficient}
Lulu Qian, David Soloveichik, and Erik Winfree.
\newblock Efficient turing-universal computation with {DNA} polymers.
\newblock In {\em DNA Computing and Molecular Programming: 16th International
  Conference, DNA 16, Hong Kong, China, June 14-17, 2010, Revised Selected
  Papers 16}, pages 123--140. Springer, 2011.

\bibitem{simmel2019principles}
Friedrich~C Simmel, Bernard Yurke, and Hari~R Singh.
\newblock Principles and applications of nucleic acid strand displacement
  reactions.
\newblock {\em Chemical reviews}, 119(10):6326--6369, 2019.

\bibitem{soloveichik2008computation}
David Soloveichik, Matthew Cook, Erik Winfree, and Jehoshua Bruck.
\newblock Computation with finite stochastic chemical reaction networks.
\newblock {\em natural computing}, 7:615--633, 2008.

\bibitem{soloveichik2010dna}
David Soloveichik, Georg Seelig, and Erik Winfree.
\newblock {DNA} as a universal substrate for chemical kinetics.
\newblock {\em Proceedings of the National Academy of Sciences},
  107(12):5393--5398, 2010.

\bibitem{szilard1929entropieverminderung}
Leo Szilard.
\newblock {\"U}ber die {E}ntropieverminderung in einem thermodynamischen
  {S}ystem bei {E}ingriffen intelligenter {W}esen.
\newblock {\em Zeitschrift f{\"u}r Physik}, 53(11-12):840--856, 1929.

\bibitem{thubagere2017cargo}
Anupama~J Thubagere, Wei Li, Robert~F Johnson, Zibo Chen, Shayan Doroudi,
  Yae~Lim Lee, Gregory Izatt, Sarah Wittman, Niranjan Srinivas, Damien Woods,
  et~al.
\newblock A cargo-sorting {DNA} robot.
\newblock {\em Science}, 357(6356):eaan6558, 2017.

\bibitem{winfree1998algorithmic}
Erik Winfree.
\newblock {\em Algorithmic self-assembly of {DNA}}.
\newblock California Institute of Technology, 1998.

\bibitem{yin2008programming}
Peng Yin, Harry~MT Choi, Colby~R Calvert, and Niles~A Pierce.
\newblock Programming biomolecular self-assembly pathways.
\newblock {\em Nature}, 451(7176):318--322, 2008.

\bibitem{yokoyama2008principles}
Tetsuo Yokoyama, Holger~Bock Axelsen, and Robert Gl{\"u}ck.
\newblock Principles of a reversible programming language.
\newblock In {\em Proceedings of the 5th Conference on Computing Frontiers},
  pages 43--54, 2008.

\bibitem{zhang2011dynamic}
David~Yu Zhang and Georg Seelig.
\newblock Dynamic dna nanotechnology using strand-displacement reactions.
\newblock {\em Nature chemistry}, 3(2):103--113, 2011.

\end{thebibliography}
